\documentclass{article}[11pt]
\usepackage[utf8]{inputenc}
\usepackage{amsmath,color}
\usepackage{amssymb}
\usepackage{amsthm}
\usepackage{thmtools}
\usepackage{thm-restate}

\usepackage{geometry}
\usepackage{xcolor}

\newtheorem{definition}{Definition}
\newtheorem{claim}{Claim}
\newtheorem{question}{Question}
\newtheorem{theorem}{Theorem}[section]

\newtheorem{remark}[theorem]{Remark}

\newtheorem{fact}[theorem]{Fact}

\newtheorem{lemma}[theorem]{Lemma}

\newtheorem{corollary}[theorem]{Corollary}

\newcommand{\R}[0]{{\ensuremath{\mathbb{R}}}}

\newcommand{\ip}[1]{\langle #1 \rangle}

\newcommand{\beq}{\begin{equation}}
\newcommand{\eeq}{\end{equation}}

\newcommand{\eps}{\varepsilon}

\newcommand{\bx}{{\bf x}}
\newcommand{\by}{{\bf y}}
\newcommand{\bw}{{\bf w}}
\newcommand{\bv}{{\bf v}}
\newcommand{\bu}{{\bf u}}
\newcommand{\bzero}{{\bf 0}}
\newcommand{\bone}{{\bf 1}}
\newcommand{\trace}{{\rm trace}}

\newcommand{\bd}[2]{\left( \begin{array}{c|c} #1 & 0 \\ \hline  0 & #2 \end{array} \right)}

\DeclareMathOperator{\E}{\mathbb{E}}
\DeclareMathOperator{\vbl}{vbl}

\newcounter{note}[section]
\renewcommand{\thenote}{\thesection.\arabic{note}}
\newcommand{\nbnote}[1]{\refstepcounter{note}$\ll${\bf Nikhil~\thenote:}
  {\sf \color{red} #1}$\gg$\marginpar{\tiny\bf NB~\thenote}}
\newcommand{\osnote}[1]{\refstepcounter{note}$\ll${\bf Ola~\thenote:}
  {\sf \color{red} #1}$\gg$\marginpar{\tiny\bf OS~\thenote}}
\newcommand{\ltnote}[1]{\refstepcounter{note}$\ll${\bf Luca~\thenote:}
  {\sf \color{red} #1}$\gg$\marginpar{\tiny\bf LT~\thenote}}

\renewcommand{\osnote}[1]{}
\renewcommand{\ltnote}[1]{}
\renewcommand{\nbnote}[1]{}

\title{New Notions and Constructions of Sparsification for Graphs and Hypergraphs}
\author{Nikhil Bansal, Ola Svensson, Luca Trevisan}

\begin{document}

\maketitle

\begin{abstract}
 A sparsifier of a graph $G$ (Bencz\'ur and Karger; Spielman and Teng) is a sparse weighted subgraph $\tilde G$ that approximately retains the same cut structure of $G$. For general graphs, non-trivial sparsification is possible only by using weighted graphs in which different edges have different weights. Even for graphs that admit unweighted sparsifiers (that is, sparsifiers in which all the edge weights are equal to the same scaling factor), there are no known polynomial time algorithms that find such unweighted sparsifiers.
 
 We study a weaker notion of sparsification suggested by Oveis Gharan, in which the number of cut edges in each  cut $(S,\bar S)$ is not approximated within a multiplicative factor $(1+\epsilon)$, but is, instead, approximated up to an additive term bounded by $\epsilon$ times $d\cdot |S| + \text{vol}(S)$, where $d$ is the average degree of the graph and $\text{vol}(S)$ is the sum of the degrees of the vertices in $S$. We provide a probabilistic polynomial time construction of such sparsifiers for every graph, and our sparsifiers have a near-optimal number of edges $O(\epsilon^{-2} n {\rm polylog}(1/\epsilon))$. We also provide a deterministic polynomial time construction that constructs sparsifiers with a weaker property having the optimal number of edges $O(\epsilon^{-2} n)$. Our constructions also satisfy a spectral version of the ``additive sparsification'' property.
 
 Notions of sparsification have also been studied for hypergraphs. Our construction of ``additive sparsifiers'' with $O_\epsilon (n)$ edges also works for hypergraphs, and  provides the first non-trivial notion of sparsification for hypergraphs achievable with $O(n)$ hyperedges when $\epsilon$ and the rank $r$ of the hyperedges are constant. Finally, we provide a new construction of spectral hypergraph sparsifiers, according to the standard definition, with ${\rm poly} (\epsilon^{-1}, r) \cdot n\log n$ hyperedges, improving over the previous spectral construction (Soma and Yoshida) that used $\tilde O(n^3)$ hyperedges even for constant $r$ and $\epsilon$.
\end{abstract}

\section{Introduction}

Bencz{\'{u}}r and Karger \cite{BK96} introduced the notion of a {\em cut sparsifier:} a weighted graph $\tilde G = (V,F)$ is an $\epsilon$ cut sparsifier of a graph $G=(V,E)$
if, for every cut $(S,V-S)$ of the set of vertices, the weighted number of cut edges in $\tilde G$ is the same as the number of cut edges in $G$, up to multiplicative error $\epsilon$, that is,
\begin{equation} 
\label{eq.bz}
\forall S\subseteq V \ \ | e_F(S) - e_E(S) | \leq \epsilon \cdot e_E(S)
\ . 
\end{equation}

\noindent where $e_X(S)$ denotes the weighted number of edges in $X$ leaving the set $S$.
A stronger notion, introduced by Spielman and Teng \cite{ST04}, is that of a {\em spectral sparsifier:} according to this notion,  a weighted graph $\tilde G = (V,F)$ is an $\epsilon$ cut sparsifier of a graph $G=(V,E)$
if 

\begin{equation} 
\label{eq.st}
\forall S\subseteq V \ \ | \bx^T L_{\tilde G} \bx - \bx^T L_G \bx | \leq \epsilon \cdot \bx^T L_G \bx \ , 
\end{equation}
\noindent where $L_X$ is the Laplacian matrix of the graph $X$.
Note that \eqref{eq.bz} is implied by \eqref{eq.st} by taking $\bx$ to be the 0/1 indicator vector of $S$. A more compact way to express \eqref{eq.st} is as
\begin{equation}
    - \epsilon L_G \preceq L_{\tilde G} - L_G \preceq \epsilon L_G \ .
\end{equation}
Batson, Spielman and Srivastava \cite{BSS12} show that, for every graph,  an $\epsilon$ spectral sparsifier (and hence also an $\epsilon$ cut sparsifier) can be constructed in polynomial time with  $O(n/\epsilon^2)$ weighted edges, which is best possible up to the constant in the big-Oh.
Sparsifiers have several applications to speeding-up graph algorithms.

For some graphs $G$, for example the ``barbell'' graph (that consists of two disjoint cliques joined by a single edge), it is necessary for a non-trivial sparsifier of $G$ to have edges of different weights. 
This has motivated the question of whether there are weaker, but still interesting, notion of sparsification that can be achieved, for all graphs, using sparsifiers that are ``unweighted'' in the sense that all edges have the same weight.

\begin{question} \label{q.unweighted}
Is a non-trivial notion of unweighted sparsification possible for all graphs?
\end{question}

Results on unweighted sparsification have focused on bounding the multiplicative error $\epsilon$ in such cases, allowing it to be superconstant \cite{AGM14,ALO15}. For graphs such as the barbell example one, however gets, necessarily, very poor bounds. But is there an alternative notion for which one can get arbitrarily good approximation on all graphs using a linear number of edges?

If one restricts this question from all graphs to selected classes of graphs, then a number of interesting results are known, and some major open questions arise.

If $G=(V,E)$ is a $d$-regular graph  such that every edge
has effective resistance $O(1/d)$,  the Marcus-Spielman-Srivastava \cite{MSS15} proof of the Kadison-Singer conjecture (henceforth, we will refer to this result as the {\em MSS Theorem})
implies that $G$ can be partitioned into almost-regular unweighted spectral sparsifiers with error $\epsilon$ and average degree $O(\epsilon^{-2})$. An interesting class of such graphs are edge-transitive graphs, such as the hypercube.

Another interesting class of graphs all whose edges have effective resistance $O(1/d)$ is the class of $d$-regular expanders of constant normalized expansion $\phi > 0$. Before the MSS Theorem, Frieze and Molloy \cite{Frieze1999} proved that such graphs can be partitioned into unweighted almost-regular graphs of average degree $O(\epsilon^{-2} \log d)$ and normalized edge expansion
at least $\phi -\epsilon$. They also show how to construct such a partition in randomized polynomial time under an additional small-set expansion assumption on $G$. Becchetti et al. \cite{BCNPT18} 
present a randomized linear time algorithm that, given a dense regular expander $G$ of degree $d = \Omega(n)$ finds an edge-induced expander in $G$ of degree $O(1)$. While both \cite{Frieze1999} and \cite{BCNPT18} find sparse expanders inside dense expanders, the work of Frieze and Molloy does not produce constant-degree graphs and the work of Becchetti et al. only applies to very dense graphs. Furthermore, neither work guarantees that one ends up with a sparse graph that is a good sparsifier of the original one. 

\begin{question} \label{q.mss}
Is there a polynomial time construction of the unweighted spectral sparsifiers of expanders whose existence follows from the Marcus-Spielman-Srivastava theorem?
\end{question}

Notions of cut sparsifiers \cite{KK15} and spectral sparsifiers \cite{SY19} have been defined for hypergraphs, generalizing the analogous definitions for graphs. In a hypergraph $H = (V,E)$, a hyperedge $e\in E$ is cut by a partition $(S, V-S)$ of the vertices if $e$ intersects both $S$ and $V-S$. As for graphs, we can define $e_E(S)$ to be the (weighted, if applicable) number of hyperdges in $E$ that are cut by $(S,V-S)$. As before, a weighted subset of edges $F$ defines a hypergraph cut sparsifier with error $\epsilon$ if

\[ \forall S \subseteq V \ \ | e_F(S) - e_E(S) | \leq \epsilon e_E(S)\,. \]
Kogan and Krauthgamer~\cite{KK15} show how to construct such a (weighted) sparsifier in randomized polynomial time using $O(\epsilon^{-2} n \cdot ( r + \log n))$ hyperedges where $r$ is the maximum size of the hyperedges which is also called the \emph{rank} of the hypergraph. 

In order to define a notion of spectral sparsification, we associate to a hypergraph $H=(V,E)$ the following analog of the Laplacian quadratic form, namely a function $Q_H$ such that
\[ Q_H (\bx) = \sum_{e\in E} w_e \cdot \max_{a,b \in e} \ \ (x_a -  x_b)^2 \]
where $w_e$ is the weight (if applicable) of hyperedge $e$. Note that with this definition we have that if $\bx = \bone_S$ for some subset $S$ of vertices then $Q_H(\bx) = e_E(S)$. Following Soma and Yoshida \cite{SY19}, we
say that a weighted hypergraph $\tilde H$ is a spectral sparsifier with error $\epsilon$ of $G$ if we have

\[ \forall \bx \in \R^V  \ \ \  | Q_{\tilde H}(\bx) - Q_{H} (\bx) | \leq \epsilon \cdot Q_H(\bx)\,. \]
Soma and Yoshida \cite{SY19} provide a randomized polynomial time construction of such sparisifiers, using $\tilde O(\epsilon^{-2} n^3)$ hyperedges.

\begin{question} Is it possible, for every hypergraph, to construct
a  weighted   spectral sparsifier  with $\tilde O_{r,\epsilon} (n)$ hyperedges?
\end{question}

As in the case of graphs, it is also natural to raise the following question.

\begin{question}
 Is a non-trivial notion of unweighted sparsification possible for all hypergraps?
\end{question}

We provide a positive answer to all the above questions.

\subsection{Our Results}

\subsubsection{Sparsification with additive error}
Oveis-Gharan suggested the following weakened definition of sparsification: if $G=(V,E)$ is $d$-regular, we say that an unweighted graph $\tilde G= (V,F)$ is an additive cut sparsifier of $G$ with error $\epsilon$ if we have
\[ \forall S\subseteq V  \ \ |c \cdot e_{ F} (S) - e_{E} (S) | \leq 2 \epsilon  d\cdot |S|\,, \]
where $c= |E|/|F|$.
Note that this (up to a constant factor change in the error parameter $\epsilon$) is equivalent to the standard notion if $G$ has constant normalized edge expansion, because $e_E (S)$ and $d\cdot S$ will then be within a constant factor of each other. On non-expanding graphs, however, this definition allows higher relative error on sparse cuts and a tighter control on expanding cuts.
(The factor of 2 has no particular meaning and it is just there for consistency with the definition that we give next for non-regular graphs.)

For non-regular graphs $G$, we say that  $\tilde G= (V,F)$ is an additive cut sparsifier of $G$ with error $\epsilon$ if we have
\[ \forall S\subseteq V  \ \ | c\cdot e_{ F} (S) - e_{E} (S) | 
\leq \epsilon \cdot ( d_{\rm avg} \cdot |S| + {\rm vol} (S) ) \]
where $c=|E|/|F|$ and $d_{\rm avg} := 2|E|/|V|$ is the average degree of $G$ and ${\rm vol} (S)$ is the volume of $S$ that is, the sum of the degrees of the vertices in $S$. In can be shown that both terms are
necessary if one wants a definition of unweighted sparsification that is applicable to all graphs.

This notion has a natural spectral analog, which we state directly in the more general form:
\[ - \epsilon \cdot (D_G + d_{\rm avg}I) \preceq c\cdot L_{\tilde G} - L_G \preceq
\epsilon \cdot (D_G + d_{\rm avg}I)\,. \]
Note, again, that if $G$ is a regular expander then this definition is equivalent to the standard definition of spectral sparsifier.

In a hypergraph, the {\em degree} of a vertex is the number of hyperedges it belongs to, and the {\em volume} of a set of vertices is the sum of the degrees of the vertices that belong to it.
With these definitions in mind, the notion of additive graph sparsifier immediately generalizes
to hypergraphs.

\subsubsection{New Graph Sparsification Constructions}

Our first result is a deterministic polynomial time construction
which achieves a weak form of unweighted additive sparsification.

\begin{theorem}[Deterministic Polynomial Time Construction] \label{th.det}
Given a graph $G=(V,E)$ and a parameter $\epsilon>0$, in deterministic polynomial time we can find a subset $F\subseteq E$ of size $|F| = O(n/\epsilon^2)$ such that, if we let $L_G = D_G - A_G$
be the Laplacian of $G$, $L_{\tilde G}= D_{\tilde G} - A_{\tilde G}$ be the Laplacian of
the graph $\tilde G = (V,F)$, $d = 2|E|/|V|$ be the average degree of $G$, and $c = |E|/|F|$, we have
\begin{equation}
2cD_{\tilde G} - 2 D_G - \epsilon D_G - \epsilon d I \preceq    c L_{\tilde G} - L_G \preceq \epsilon D_G + \epsilon d I
   \label{eq.det}
\end{equation}
\end{theorem}

Note, in particular, that we get that for every set
of vertices $S\subseteq V$ we have
\begin{equation}
    - \epsilon |E| \leq c e_F(S) - e_E(S) \leq \epsilon \text{vol}(S) + \epsilon d |S|
 \end{equation}
The first inequality follows by computing the quadratic forms of \eqref{eq.det}  with the $\pm 1$ indicator vector $\bx:= \bone_S - \bone_{\bar S}$ of $S$, and noting that $\bx^T L_G \bx = 4 e_E(S)$, that $\bx^T L_{\tilde G} \bx = 4 e_F(S)$, that
\[ \bx^T M \bx = {\rm trace} (M) \]
for every diagonal matrix $M$, and that $\trace(D_G) = \trace(cD_{\tilde G}) = \trace(dI) = 2|E|$.
The second inequality follows by computing the quadratic forms of \eqref{eq.det} with the $0/1$ indicator vector $\bx = \bone_S$ of $S$, and
noting that $\bx^T L_G \bx =  e_E(S)$,  $\bx^T L_{\tilde G} \bx =  e_F(S)$, $\bx ^T I \bx = |S|$ and $\bx^T D_G \bx = {\rm vol}(S)$.

Our proof is based on the online convex optimization techniques of Allen-Zhu, Liao and Orecchia \cite{ALO15}. The construction of \cite{ALO15} involves weights for two reasons: one reason is a change of basis that maps $L_G$ to identity, a step that is not necessary in our setting and that could also be avoided in their setting if $G$ is a graph all whose edges have bounded effective resistance. The second reason is more technical, and it is to avoid blowing up the ``width'' on the online game that they define. The second issue comes up when one wants to prove $c L_{\tilde G} - L_G \succeq - \epsilon \cdot (D_G+ dI)$, but is not a problem for the upper bound
$c L_{\tilde G} - L_G \preceq  \epsilon\cdot (D_G + dI)$.

To sidestep this problem, we set the goals of proving the bounds
\[ c L_{\tilde G} - L_G \preceq \epsilon \cdot ( dI + D_G) \]
\[ c SL_{\tilde G} - SL_G \preceq \epsilon \cdot ( dI + D_G) \]
where $SL_G$ denotes the {\em signless Laplacian} of a graph $G$, defined as $D_G + A_G$. Note that the above PSD inequalities are equivalent to \eqref{eq.det}. 

The reasons why, when our goal is the PSD inequalities above, we are able to control the width without scaling (and without weighing the edges)  are quite technical, and we defer further discussion to Section \ref{sec.det}. 

Our next result is a probabilistic construction of sparsifiers with additive error matching the Oveis-Gharan definition.

\begin{theorem}[Probabilistic Polynomial Time Construction] \label{th.prob.graph}
Given an $n$-vertex graph $G=(V,E)$ and a parameter $\epsilon>0$, in probabilistic polynomial time we can find a subset $F\subseteq E$ of size\footnote{where $\tilde O(\cdot )$ hides $\log(1/\eps)^{O(1)}$ factors} $|F| = n \cdot \tilde O(1/\epsilon^2)$ such that, if we let $L_G = D_G - A_G$
be the Laplacian of $G$, $L_{\tilde G}= D_{\tilde G} - A_{\tilde G}$ be the Laplacian of
the graph $\tilde G = (V,F)$, $d = 2|E|/|V|$ be the average degree of $G$, and $c = |E|/|F|$, we have
\begin{equation}
\label{eq.th.prob.graph}
 - \epsilon D_G - \epsilon d I \preceq    c L_{\tilde G} - L_G \preceq \epsilon D_G + \epsilon d I\,.
\end{equation}
\end{theorem}
When we apply the above result to a $d$-regular expander $G$, we obtain a graph $\tilde G$ whose average
(and maximum) degree is $\tilde O(\epsilon^{-2})$ and which is itself a good expander. More precisely, if $G$ has normalized edge expansion $\phi$ and $\tilde G$ is as above, then the normalized edge expansion of $\tilde G$ is about $\phi - 2\epsilon$. Recall that Frieze and Molloy can find a $\tilde G$ as above but with degree $O(\epsilon^{-2} \log d)$ rather than $O(\epsilon^{-2} {\rm polylog} \epsilon^{-1})$. Furthermore, if $G$ is a $d$-regular expander of normalized edge expansion $\phi$, we have\footnote{There is some abuse of notation in \eqref{eq.additive.expander}, because \eqref{eq.additive.expander} only holds in the space orthogonal to $\bone = (1,1,\cdots,1)$.} that
\begin{equation} D_G + dI = 2dI \preceq O(\phi^{-2} L_G)
\label{eq.additive.expander}
\end{equation}
and so the unweighted sparsifier $\tilde G$ of $G$ given by the above theorem is also a spectral sparsifier in the standard sense. This answers Questions 1 and 2 of the previous section.

We briefly discuss the techniques in the proof. Following Frieze and Molloy~\cite{Frieze1999} and  Bilu and Linial \cite{Bilu2006}, we apply the Lov\'{a}sz Local Lemma~\cite{LLL} (LLL) to construct an additive cut sparsifier. One difficulty with this approach is that one has to verify that the sparsifier approximates each of the exponentially many cuts. Indeed, if one defines a ``bad'' event for each one of these cuts, there are too many  events that are dependent in order to successfully apply LLL. A key insight in~\cite{Frieze1999} is that it is sufficient to verify those cuts $(S, V - S)$ where $S$ induces a connected subgraph. This makes a big difference in graphs of maximal degree $d \ll n$: for a vertex $v$, there are $\approx n^{\ell-1}$ subsets of $\ell$ vertices containing $v$ whereas one can prove that there are at most $\binom{d(\ell-1)}{\ell-1}$ such subsets of size $\ell$ that induce a connected  subgraph. This allows one to manage the exponentially many events and get almost optimal results with LLL. Indeed, we obtain a close to optimal average degree $\tilde O(\epsilon^{-2})$. This improves upon the average degree bound in $\epsilon^{-2} \log d$ \cite{Frieze1999}  . We achieve this by an iterative procedure that intuitively halves the number of edges, instead of sparsifying the graph ``in one go.''

Another difference is that, in contrast to~\cite{Frieze1999} and~\cite{Bilu2006}, we can use recent constructive versions of LLL~\cite{Haeupler11} to give an efficient probabilistic time algorithm for finding the sparsifier.  To apply the constructive version of LLL in the presence of exponentially bad events, one needs to find a subset of bad events  of polynomial size such that the probability that any other bad event is true is negligible. We show that this can be achieved by selecting the subset of events corresponding to cuts $(S, V - S)$ so that $S$ induces a connected graph and $|S| =O(\log_{d}(n))$.  This gives us an efficient probabilistic algorithm for finding a cut sparsifier which we also generalize to hypergraphs (as we state in the next section). For graphs, we then adapt the techniques of Bilu and Linial~\cite{Bilu2006} to go from a cut sparsifier  to a spectral one. To do so we need  to consider some more bad events in the application of LLL than needed by Bilu-Linial who worked with ``signings'' of the adjacency matrix. Specifically, in addition to the events that they considered, we need to also bound the degree of vertices. 

\subsubsection{New Hypergraph Sparsification Constructions}

\begin{theorem}[Hypergraph cut sparsification with additive error] \label{th.prob}
Given an $n$-vertex hypergraph $H=(V,E)$ of rank $r$  and a parameter $\epsilon>0$, in probabilistic  polynomial time we can find a subset $F\subseteq E$ of size $|F| =O\left(  \frac nr  \cdot \frac 1{\epsilon^2} \log \frac r\epsilon \right)$ such that, if we let $d = r|E|/|V|$ be the average degree of $H$, and $c = |E|/|F|$, the following holds with probability at least $1-n^{-2}$:
\begin{equation}
|    c e_{F} (S) - e_E (S) | \leq  \epsilon d |S| + \epsilon {\rm vol}(S) \qquad \forall S \subseteq V\,.
\end{equation}
\end{theorem}

The proof follows the same approach as the first part of our proof of Theorem \ref{th.prob.graph}, and in fact we present directly the proof for hypergraphs, leaving the result for graphs as a corollary. It might seem strange that the number of hyperedges in our sparsifier is, for fixed $\epsilon$, of the form $O\left( \frac nr \log r\right)$, since, intuitively, the sparsification problem should only become harder when $r$ grows. The reason is that, even in a regular hypergraph,  $d|S|$ overestimates the number of hyperedges incident on $S$ by up to a factor of $r$, and so, in order to have a non-trivial guarantee, one has to set $\epsilon < 1/r$.

\begin{theorem}[Hypergraph sparsification with multiplicative error]
There is a randomized polynomial time algorithm that, given a hypergraph of rank $r$, finds a weighted spectral sparsifier with multiplicative error $\epsilon$ having $O(\epsilon^{-2} r^3 n \log n)$ hyperedges.
\end{theorem}

The above result should be compared with the $O(\epsilon^{-2} n^3\log n)$ hyperedges of the
construction of Soma and Yoshida \cite{SY19}. Our approach is to provide an ``hypergraph analog'' of the
spectral graph sparsifier construction of Spielman and Srivastava \cite{SS11}. Given $H$, we
construct an associated graph $G$ (in which each hyperedge of $H$ is replaced by a clique in $G$), we compute the effective resistances of the edges of $G$, and we use them to associate a notion of ``effective resistance'' to the hyperedges of $H$. Then we sample from the set of hyperdedges of $H$ by letting the sampling probability of each hyperedge be proportional to its ``effective resistance'' and we weigh them so that the expected weight of each hyperedge in the sample is the same. At this point, to bound the error, Spielman and Srivastava complete the proof by applying a matrix concentration bound for the spectral norm of sums of random matrices. For hypergraphs, we would like to have a similar concentration bound on the error given by, 
\beq \max_{\bx \in \R^V : \| \bx \|=1 }\ \ \  \sum_{e\in H} (1-W_e) \cdot \max_{a,b \in e}  \ \  ( x_a - x_b )^2
\label{eq:hyp-intro}
\eeq
where $W_e$ is a random variable that is 0 if the hyperedge $e$ is not selected and it is its 
weight in the sparsifier if it is selected, with things set up so that $1-W_e$ has expectation zero.
(Actually, this would only lead to a sparsifier with additive error: to achieve multiplicative error we have to study an expression such as the one above but after a change of basis defined in terms of the associated graph. For simplicity we will ignore this point in this overview.) 

However, unlike in the graph case, the expression in \eqref{eq:hyp-intro} does not correspond to the spectral norm, or any other standard linear-algebraic norm, due to the $\max$ term, and the key difficulty in all previous approaches to the problem was to get suitable upper bounds on this quantity.
Our main idea is to consider the quantity $V_x = \sum_{e\in H} (1-W_e) \cdot \max_{a,b \in e}  \ \  ( x_a - x_b )^2$ and view it as a random process indexed by the set of  all unit vectors $x$, and directly argue about its supremum over all such $x$, using the technique of generic chaining.  In particular, we relate the metric given by the sub-gaussian norm of the increments of the process $V_x$  to another suitably defined Gaussian random process on the associated graph $G$ of $H$, which is much easier to analyze. This allows us to relate the bound on the supremum of $V_x$ to a related expression on the graph $G$, for which we can use known matrix concentration bounds.


\section{Preliminaries}

\subsection{Linear Algebra Preliminaries}

In this paper all matrices will have real-valued entries.

A matrix $M$ is Positive Semidefinite (abbreviated PSD and written $M \succeq \bzero$) if it is symmetric and all its eigenvalues are non-negative. Equivalently, $M$ is PSD if and only if 
\[ \forall \bx \in \R^n , \ \ \bx ^T M \bx \geq 0 \]
that is, the quadratic form of $M$ is always non-negative. The {\em trace} of a matrix is the sum of its diagonal entries. For a symmetric matrix, its trace is equal to the sum of its eigenvalues, counted with multiplicities. A {\em density matrix} is a PSD matrix of trace one. The operator norm of a matrix $M$ is
\[ \| M\| = \sup_{\bx : \| \bx \|_2 = 1} \| M x \|_2 \]
If $M$ is symmetric, then the above is the largest absolute value of the eigenvalues of $M$ and we also refer to it as the spectral norm or spectral radius of the matrix.

If $A$ and $B$ are matrices of the same size, then their Frobenius inner product is defined as
\[ \langle A, B \rangle = \sum_{i,j} A_{i,j} B_{i,j} = \trace (A^T B) \]
and we will also sometimes denote it as $A \bullet B$. Note that if $M$ is a symmetric matrix we have
\[ \| M \| = \sup_{X \ {\rm density\ matrix}} \ \ \ | M \bullet X | \]

If $M$ is a symmetric $n\times n$ matrix with spectral decomposition

\[ M = \sum_{i=1}^n \lambda_i \bv_i \bv_i ^T \ ,  \]
then the ``absolute value'' of $M$ is the PSD matrix
\[ |M| = \sum_{i=1}^n | \lambda_i | \bv_i \bv_i ^T  \ . \]

\subsection{Reduction to bounded-degree case}
\label{sec.reduction}

We show that, in proving Theorem \ref{th.det}, Theorem \ref{th.prob.graph} and Theorem \ref{th.prob}, it is enough to prove weaker bounds where $d_{\rm avg} I + D_G$ is replaced by $d_{\rm max} \cdot I$, and ${\rm vol} (S)$ is replaced by $d_{\max} |S|$, where $d_{\rm max}$ is the maximum degree.

Consider the following construction: given a graph $G=(V,E)$ of average degree $d_{\rm avg} = 2|E|/V$, construct a new graph $G' = (V'E')$ such that
\begin{itemize}
    \item To each node $v\in V$ there  corresponds a cloud of $\lceil d_v/\lceil d_{\rm avg}\rceil \rceil$ nodes in $V$.
    \item To each edge $(u,v) \in E$ there corresponds an edge in $E'$ between the cloud of $u$ and the cloud of $v$.
    \item Each vertex in $G'$ has degree at most $d'_{\rm max} = \lceil d_{\rm avg}\rceil $.
\end{itemize}

A construction satisfying the above property can be realized by replacing the vertices of $V$, in sequence, by a cloud as required, and then replacing $v$ in the edges incident to $v$ by vertices in the cloud of $v$, in a balanced way.

Now suppose that $F' \subseteq E'$ is a subset of the edges of $G'$ and that $F\subseteq E$ is the set of edges of $G$ corresponding to the edges of $F'$. Let $\tilde G$ be the graph $\tilde G=(V,F)$ and $\tilde G' = (V,F')$. Let $\bx \in \R^V$ be any vector, and define $\bx' \in R^V$ to be the vector such that $\bx'_{v'} = \bx_v$ if $v'$ is in the cloud of $v$. Then we observe that

\[ \bx^T L_G \bx = \bx'^T L_{G'} \bx' \]
\[ \bx^T L_{\tilde G} \bx = \bx'^T L_{\tilde G'} \bx' \]
\[ \bx'^T  (d'_{\rm max}I ) \bx'
\leq \bx^T (\lceil d_{\rm avg} \rceil + D_G) \bx 
\]
The only non-trivial statement is the third one. To verify it, we see that the left-hand side is

\[ \bx'^T  (d'_{\rm max}I ) \bx'
= \sum_{v\in V} \left\lceil \frac{d_v}{\lceil d_{\rm avg}\rceil } \right\rceil 
\cdot \lceil d_{\rm avg} \rceil  \bx_v^2  \leq \sum_{v\in V} (d_v + \lceil d_{\rm avg} \rceil ) \cdot \bx_v^2 \]
This means that we can start from an arbitrary graph $G$, construct $G'$ as above, find an unweighted sparsifier $\tilde G' = (V', F')$ of $G'$, and
then obtain a set $F$ of edges such that
$\tilde G = (V,F)$ is an unweighted sparsifier for $G$, with the property that any bound dependent on $d'_{\rm max} I$ on the quality of the sparsification of $\tilde G'$ becomes a bound in terms of $(\lceil d_{\rm avg} \rceil + D_G)$ (and we can drop the ceiling at the cost of a constant factor in the error).

If $H= (V,E)$ is a hypergraph we can similarly construct a hypergraph $H'= (V',E')$ such that
\begin{itemize}
    \item To each node $v\in V$ there  corresponds a cloud of $\lceil d_v/\lceil d_{\rm avg}\rceil \rceil$ nodes in $V$.
    \item To each edge $(u,v) \in E$ there corresponds an edge in $E'$ between the cloud of $u$ and the cloud of $v$.
    \item Each vertex in $H'$ has degree at most $d'_{\rm max} = \lceil d_{\rm avg}\rceil $.
\end{itemize}

Similarly to the graph case, for every set $S\subseteq V$ we can define a set $S' \subseteq V'$ (the union of the clouds of vertices in $S'$) and for every set
$F'\subseteq E'$ we can define
a set of hyperedges $F\subseteq E$ of the same cardinality such that
\[ e_{E} (S) = e_{E'} (S') \]
\[ e_{F}(S) = e_{F'} (S') \]
\[ d'_{\max} |S| \leq
\lceil d_{\rm avg} \rceil |S| + 
{\rm vol}_{H} (S) \]

We also note that, in both constructions, the maximum degree and the average degree of the new graph (or hypergraph) are within a constant factor.

\section{Deterministic Construction}
\label{sec.det}
In this section we use the online convex optimization approach of Allen-Zhu, Liao and Orecchia \cite{ALO15} to construct a weak form of unweighted additive spectral sparsifiers, and we prove Theorem \ref{th.det}. 
Given the reduction described in Section \ref{sec.reduction}, it is enough to prove
the following theorem.

\begin{theorem} \label{th.det.bounded}
There is a deterministic polynomial time algorithm
that given a graph $G= (V,E)$ of maximum degree
$d_{\max}$ and a parameter $\epsilon$ outputs
a multiset $F$ of $O(|V|/\epsilon^2)$ edges such 
that the graph $\tilde G = (V,F)$ satisfies
\[ 2 c D_{\tilde G} - 2 D_G - \epsilon d_{\max}I
\preceq c L_{\tilde G} - L_G \preceq \epsilon d_{\max}I
\]
where $c = |E|/|F|$.
\end{theorem}

We are interested in the following online optimization setting: at each time $t=1,\ldots$, an {\em algorithm} comes up with a solution $X_t$, which is an $n\times n$ density matrix, and an adversary comes up with a cost matrix $C_t$, which is an $n\times n$ matrix, and the algorithm receives a payoff $X_t \bullet C_t$. The algorithm comes  up with $X_t$ based on knowledge of $X_1,\ldots,X_{t-1}$ and of $C_1,\ldots,C_{t-1}$, while the adversary comes up with $C_t$ based on $X_1,\ldots,X_{t}$ and on $C_1,\ldots,C_{t-1}$. The goal of the algorithm is to maximize the payoff. After running this game for $T$ steps, one defines the {\em regret} of the algorithm as
\[R_T := \left( \sup_{X \ {\rm density\ matrix}} \ \ \ \sum_{t=1}^T X \bullet C_t \right) -  \left( \sum_{t=1}^T X_t \bullet C_t \right)\,. \]

\begin{theorem}[Allen-Zhu, Liao, Orecchia~\cite{ALO15}]
There is a deterministic polynomial algorithm that, given a parameter $\eta>0 $, after running for $T$ steps against an  adversary
that provides cost matrices $C_t$ restricted as described below, achieves a regret bound
\[ R_T \leq O(\eta) \cdot \sum_{t=1}^T (X_t \bullet |C_t| ) \cdot \| X_t^{1/4} C_t X_t ^{1/4} \| + \frac {2 \sqrt n}{\eta} \ . \]
Furthermore, if $C_t$ is block-diagonal, then $X_t$
is also block-diagonal with the same block structure
The restrictions on the adversary are that at each step $t$ the cost function $C_t$
is  positive semidefinite or negative semidefinite and satisfies
\[ \eta X_t^{1/4} C_t X_t ^{1/4} \preceq  \frac I 4 \,. \]
\label{th.alo} 
\end{theorem}
\begin{remark} The theorem above is the $q=2$
case of Theorem 3.3 in \cite{ALO15}. The Furthermore part is not stated explicitly in \cite[Theorem 3.3]{ALO15} but can be verified by inspecting the proof.
Note that what we are calling $C_t$ corresponds to $-C_t$ in the treatment of  \cite{ALO15}, which is why their cost minimization problem becomes a maximization problem here, and the condition that $C_t$  satisfy $\eta X_t^{1/4} C_t X_t ^{1/4} \succeq  - \frac I 4$ becomes the  condition that we have in the above theorem.
\end{remark}
To gain some intuition about the way we will use the above theorem, note that the definition of regret implies that we have

\[ \lambda_{\max} \left( \sum_{t=1}^T C_t \right) = R_T + \sum_{t=1}^T  C_t \bullet X_t\,, \]
where $\lambda_{\max}(\cdot)$ denotes the largest eigenvalue of the matrix.
Now suppose that we play the role of the adversary against the algorithm of Theorem \ref{th.alo}, and that, at time $t$, we reply to the solution $X_t$ of the algorithm with a cost matrix of the form $mL_{a_t, b_t} - L_G$ where $m:=|E|$ and  $(a_t,b_t)$ is an edge chosen so that
\[ X_t \bullet (mL_{a_t,b_t} - L_G) \leq 0 \]
We know that such an edge $(a_t,b_t)$ must exist, because the average of the left-hand side above is zero if we compute it for a uniformly chosen random $(a_t,b_t) \in E$. After playing this game for $T$ steps we have
\[ \lambda_{\max} \left( m \sum_{t=1}^T L_{a_t,b_t} - T L_G\right)
\leq R_T \]
and, calling $F$ the multiset $\{ (a_t,b_t) : t=1,\ldots,T \}$,
 calling $\tilde G = (V,F)$ the multigraph of such edges and
 $c = |E|/|F| = m/T$, and noting that $L_{\tilde G} = \sum_t L_{a_t,b_t}$ we have
 \[  c \cdot  L_{\tilde G} -  L_G \preceq \frac 1T R_T  \cdot I\]
which, provided that we can ensure that $R_T$ is small, is one side of the type of bounds that we are trying to prove. 

In order to get a two-sided bound, one would like to use the idea that
\[ \lambda_{\max} \bd{M}{-M} = \| M \| \]
and play the above game using, at step $t$, a cost matrix of the form
\[ C_t = \bd{mL_{a_t,b_t} - L_G}{L_G - mL_{a_t,b_t}} \]
where the edge $(a_t,b_t)$ is chosen so that
\[ X_t \bullet C_t \leq 0 \]
Then, if we define $c$ and $\tilde G$ as above, we would reach the conclusion 
\[ \| cL_{\tilde G} - L_G \| \leq \frac 1T R_T \]
and what remains to do is to see for what value of $T$ we get a sufficiently small regret bound.

Unfortunately this approach runs into a series of difficulties. 

First of all, our cost matrix is neither positive semidefinite nor negative semidefinite. 

We could make it positive semidefinite by shifting, that is, by adding a multiple of the identity. This is not a problem for the block $mL_{a_t,b_t} - L_G$, whose smallest eigenvalue is at most $2d_{\max}$ in magnitude, but it is a serious problem for the block $L_G - mL_{a_t,b_t}$, whose smallest eigenvalue is of the order of $-m$: the shift needed to make this block PSD would be so big that the terms $X_t \bullet |C_t|$ in the regret bound would be too large to obtain any non-trivial result.

Another approach, which is closer to what happens in \cite{ALO15}, is to see that the analysis of Theorem \ref{th.alo} applies also to block-diagonal matrices in which each block is either positive semidefinite or negative semidefinite. This way, we can shift the two blocks in different directions by $2d_{\max} I$ and get the cost function in a form to which Theorem \ref{th.alo} applies, but then we would still be unable to get any non-trivial bound because the term $X_t\bullet |C_t|$
could be in the order of $m$, while the analysis requires that term to be of the order of $d_{\max}$ to get the result we are aiming for. 
To see why, note that if $C_t$ is a block-diagonal matrix with a positive semidefinite block and a negative semidefinite block, then $|C_t|$ is just the same matrix except that the negative semidefinite block appears negated. Recall that we wanted to select an edge so that 
$X \bullet C_t$ is small:  what will happen is that the PSD block gives a positive contribution, the NSD block gives a negative contribution, and $X\bullet |C_t|$ is the sum of the absolute values of these contributions, which can both be order of $m$.

We could work around this problem by scaling the matrix in a certain way, but this would make the analysis only work for a weighted sparsifier. This difficulty is the reason why \cite{ALO15} construct a weighted sparsifier even if the effective resistances of all the edges of $G$ are small, a situation in which an unweighted sparsifier is known to exist because of the Marcus-Spielman-Srivastava theorem.

We work around these difficulties by reasoning about the {\em signless} Laplacian. If $G$ is a graph with diagonal degree matrix $D_G$ and adjacency matrix $A_G$, then the signless Laplacian of $G$ is defined as the matrix $D_G + A_G$. We denote by $SL_G$ the signless Laplacian of a graph $G$, and by $SL_{a,b}$ the signless Laplacian of a graph containing only the single edge $(a,b)$. Equation \eqref{eq:reg-term1}
below shows that, in this case, the term $X_t \bullet |C_t|$ in the regret bound can be bounded in term of $d_{\max}$ and are never order of $m$.

Recall that, like the Laplacian, the signless Laplacian is a PSD matrix whose largest eigenvalue is at most $2d_{\max}$.

To prove Theorem \ref{th.det.bounded}, we will play the role of the adversary against the algorithm of Theorem \ref{th.alo} with the PSD cost matrix

\[ C_t := 2d_{\max} I + \bd{mL_{a_t,b_t} - L_G}{mSL_{a_t,b_t} - SL_G} \]
where the edge $(a_t,b_t)$ is chosen so that
\[ X_t \bullet \bd{mL_{a_t,b_t} - L_G}{mSL_{a_t,b_t} - SL_G} \leq 0\,. \]
Since $X \bullet I = 1$ for every density matrix, we get that, after $T$ steps, if we define $F$ to be the multiset of selected edges, $c = \frac{|E|}{|F|} = \frac mT$, and $\tilde G = (V,F)$, then we have
\[ c L_{\tilde G} - L_G \preceq \frac {R_T}T \cdot I \]
\[ c SL_{\tilde G} - SL_G \preceq \frac {R_T}T \cdot I \]
and so it remains to show that we can make $R_T \leq \epsilon d_{\max}\cdot T$ by choosing $T= O(n/\epsilon^2)$.

Let us analyze the quantities that come up in the statement of Theorem \ref{th.alo}.

Since $C_t$ is PSD, we have
\beq
 X_t \bullet |C_t| = X_t \bullet C_t
= 2d_{\max} + X_t \bullet \bd{mL_{a_t,b_t} - L_G}{mSL_{a_t,b_t} - SL_G}  \leq 2 d_{\max}\,. 
\label{eq:reg-term1}
\eeq

The non-trivial part of the analysis is the following bound.
\begin{claim} At every time step $t$ we have
\begin{equation} \label{eq.det.claim}
\| X_t ^{1/4} C_t X_t^{1/4}\| \leq O(\sqrt {d_{\max} \cdot m})
\end{equation}
\end{claim}

\begin{proof}
Recall from Theorem~\ref{th.alo} that matrices $X_t$ will have the same block structure as the cost matrices $C_t$. 
We can therefore write the matrix $X_t$ as
\[ X_t = 
\left( 
  \begin{array}{c|c}
    Y_t & 0\\
    \hline
    0 & Z_t
  \end{array}
\right)
\] 
Then
\[ X^{1/4}_t = \bd{Y_t^{1/4}}{Z_t^{1/4}}
\]
and
\[
\| X_t ^{1/4} C_t X_t^{1/4}\|
= \max \{ \| Y_t ^{1/4} (mL_{a_t,b_t} +2d_{\max}I - L_G) Y_t^{1/4}\|,
\| Z_t ^{1/4} (mSL_{a_t,b_t} +2d_{\max}I - SL_G) Z_t^{1/4}\| \}
\]
Using the triangle inequality and the fact
that all the eigenvalues of $X_t$, and hence of $Y_t$, of $Z_t$, of $Y_t^{1/4}$ and $Z_t^{1/4}$ are at most one, 
we have
\[
\| Y_t ^{1/4} (mL_{a_t,b_t} +2d_{\max}I - L_G) Y_t^{1/4}\|
\leq m \| Y_t ^{1/4} L_{a_t,b_t} Y_t^{1/4}\| + 2 d_{\max} 
\]
\[
\| Z_t ^{1/4} (mSL_{a_t,b_t} +2d_{\max} I - SL_G) Y_t^{1/4}\|
\leq m \| Z_t ^{1/4} SL_{a_t,b_t} Z_t^{1/4}\|
+ 2 d_{\max} 
\]
Also recall that we chose $(a_t,b_t)$ so that we would have
\[ X_t \bullet \bd{mL_{a_t,b_t} - L_G}{mSL_{a_t,b_t} - SL_G} \leq 0 \]
which is the same as
\[ Y_t \bullet m L_{a_t,b_t} 
+ Z_t \bullet m SL_{a_t,b_t}
\leq Y_t \bullet L_G + Z_t \bullet SL_G = X_t \bullet \bd{L_G}{SL_G} \leq
\lambda_{\max} \bd{L_G}{SL_G}  \leq 
2d_{\max} \]
which implies
\[ Y_t \bullet m L_{a_t,b_t}  \leq 2d_{\max} \]
\[ Z_t \bullet m SL_{a_t,b_t} \leq 2d_{\max} \]
Now let us write
\[ Y_t = \sum_i \lambda_i \by_i \by_i^T \]
where $\lambda_i$ are the eigenvalues of $Y_t$
and $\by_i$ are a orthonormal basis of eigenvectors of $Y_t$, and let us also write
\[ mL_{a_t,b_t} = \bw \bw^T \]
where $\bw$ is the vector $\sqrt m \cdot (\bone_{a_t} - \bone_{b_t})$ of length $\sqrt {2m}$.
Then
\[ \| Y_t^{1/4} m L_{a_t,b_t} Y_t^{1/4}\|
= \| Y_t^{1/4} \bw \|^2 = \bw^T Y_t^{1/2} \bw 
= \sum_i \sqrt{\lambda_i } \bw^T \by_i \by_i^T \bw
= \sum_i \sqrt{\lambda_i} \langle \bw, \by_i \rangle^2
\]
Finally, by Cauchy-Schwarz,
\[ \sum_i \sqrt{\lambda_i} \langle \bw, \by_i \rangle^2
\leq \sqrt{\sum_i \langle \bw , \by_i \rangle^2}
\cdot \sqrt{\sum_i \lambda_i \langle \bw , \by_i \rangle^2} = \| \bw \| \cdot \sqrt{\bw^T Y_t \bw}
\leq \sqrt{2m} \cdot \sqrt{2d_{\max}} \]
In a completely analogous way we can prove that
\[ \|Z_t^{1/4} m SL_{a_t,b_t} Z_t^{1/4} \| \leq 2 \sqrt{d_{\max} m} \]
\end{proof}

To conclude the proof, take $\eta$ such that
\[ \eta \| X^{1/4} C_t X^{1/4} \| \leq \min \{ 1/4, \epsilon \} \]
which, by the above claim, means that it can be done
by choosing $\eta = \epsilon /O(\sqrt{d_{\max} m })$. Then using \eqref{eq:reg-term1} and that $m\leq d_{max} n$
we have the regret bound
\begin{eqnarray*}
R_T & \leq &   O(\eta) \cdot T \cdot 2d_{\max} \cdot O(\sqrt{d_{\max} \cdot m}) + \frac{2 \sqrt{n}}{\eta} \\
&  = & O(\epsilon \cdot T \cdot d_{max} ) + O\left(\frac {\sqrt{d_{\max} m n } } {\epsilon } \right) \leq O(\epsilon \cdot T \cdot d_{max} ) + O\left(\frac {d_{\max} n} {\epsilon } \right)
\end{eqnarray*}
When $T=O(n/\epsilon^2)$, the above upper bound is $O(\epsilon \cdot T \cdot d_{\max})$, which means that we have constructed a graph $\tilde G$ with $T = O(n/\epsilon^2)$ edges such that
\[ \frac mT L_{\tilde G} - L_G \preceq O(\epsilon) \cdot d_{\max} \cdot I  \]
\[ \frac mT SL_{\tilde G} - SL_G \preceq O(\epsilon) \cdot d_{\max} \cdot I \]
where the second equation is equivalent to
\[ \frac mT L_{\tilde G} - L_G \succeq \frac{2m}T D_{\tilde G} -2D_G - O(\epsilon) \cdot d_{\max} \cdot I  \]
proving Theorem \ref{th.det.bounded}.

\section{Probabilistic  construction of additive sparsifiers}
\label{sec.prob}

In this section, we give probabilistic algorithms for constructing additive spectral sparsifiers of hypergraphs. 
Specifically, we prove the following theorem which, by the reduction in Section~\ref{sec.reduction}, implies Theorem~\ref{th.prob}. That we can choose the normalization constant $c$ to equal $|E|/|F|$ in Theorem~\ref{th.prob} is because, in the reduction, the following theorem is used for a graph where $d_{\max}$ approximately equals the average degree.
\begin{theorem}\label{th.probmax} 
  Given an $n$-vertex hypergraph $H=(V,E)$ of rank $r$ and of maximal degree $d_{\max}$  together with a parameter $\epsilon>0$, in probabilistic  polynomial time we can find a subset $F\subseteq E$ of size $|F| =O\left(  \frac nr  \cdot \frac 1{\epsilon^2} \log \frac r\epsilon \right)$ such that, if we let  $c$ be a normalization constant, the following holds with probability at least $1-n^{-2}$:
\begin{equation}
  |    c\cdot  e_{F} (S) - e_E (S) | \leq  \epsilon d_{\max} |S| \qquad \forall S\subseteq V\,.
\end{equation}
\end{theorem}
In Section~\ref{sec:additive_spectral} we then generalize the techniques for simple graphs to obtain additive spectral sparsifiers as stated in Theorem~\ref{th.prob.graph}.

Our arguments are inspired by those used by Frieze and Molloy~\cite{Frieze1999} and subsequently by Bilu and Linial~\cite{Bilu2006}. They use the Lov\'{a}sz Local Lemma (LLL)~\cite{LLL} with an exponential number of bad events and may at first seem non-constructive.   However, rather recent results give efficient probabilistic algorithms even in these applications of LLL. Theorem~3.3 in~\cite{Haeupler11} will be especially helpful for us. To state it we need to introduce the following notation. We let $\mathcal{P}$ be a finite collection of mutually independent random variables $\{P_1, P_2, \ldots, P_n\}$ and let $\mathcal{A} = \{A_1, A_2, \ldots, A_m\}$ be a collection of events, each determined by some subset of $\mathcal{P}$. For any event $B$ that is determined by a subset of $\mathcal{P}$ we denote the smallest such subset by $\vbl(B)$. Further, for two events $B$ and $B'$ we write $B\sim B'$ if $\vbl(B) \cap \vbl(B') \neq \emptyset$. In other words, $B$ and $B'$ are neighbors in the standard dependency graph considered in LLL.   Finally, we say that a subset $\mathcal{A}' \subseteq \mathcal{A}$ is an efficiently verifiable core subset if there is a polynomial time algorithm for finding a true event in $\mathcal{A}'$ if any.  We can now state a (slightly) simplified version of Theorem 3.3 in \cite{Haeupler11} as follows:
\begin{theorem}
  Let $\mathcal{A}' \subseteq \mathcal{A}$ be an efficiently verifiable core subset of $\mathcal{A}$. If there is an $\eps \in [0,1)$ and an assignment of reals $x: \mathcal{A} \rightarrow (0,1)$ such that:
    \begin{align}
      \label{eq:cond}
      \forall A \in \mathcal{A}: \Pr[A] \leq (1-\eps) x(A) \prod_{B\in \mathcal{A}: B \sim A} (1- x(B))\,,
    \end{align}
    then there exists a randomized polynomial time algorithm that outputs an assignment in which all events in $\mathcal{A}$ are false with probability at least $1- \sum_{A \in \mathcal{A} \setminus \mathcal{A}'} x(A)$.
  \label{thm:LLL_constructive}
\end{theorem}

The following lemma says that we can roughly half the degree of vertices without incurring too much loss in the cut structure. Applying this lemma iteratively then yields a sparsifier. We use the following notation: For an edge set $X$ and disjoint vertex subsets $S$ and $T$, we let $\delta_X(S,T)$ denote the set of edges with one endpoint in $S$ and one in $T$; for brevity, we also write $\delta_X(S)$ for $\delta_X(S, \bar S)$. Also recall that $e_X(S,T) = |\delta(S,T)|$ and  $e_X(S) = |\delta_X(S)|$.
\begin{lemma}
  There exists a probabilistic polynomial-time algorithm that, given an $n$-vertex hypergraph $H = (V,E)$ of maximal degree $d_{\max}$ and of rank $r$, outputs a subgraph $\tilde{H} = (V,F)$ with $F\subseteq E$ such that the following holds with probability at least $1-n^{-3}$:
  \begin{align*}
    \left|2\cdot e_F(S) - e_E(S) \right| \leq 10 \sqrt{d \log (dr)}\cdot |S| \qquad \mbox{ for every $S\subseteq V$.}
  \end{align*}
  \label{lemma:cut_iteration}
\end{lemma}
\begin{proof}
  Throughout the proof we let $d= d_{\max}$.
  The proof adapts the arguments in~\cite{Bilu2006} (which in turn are similar to those in~\cite{Frieze1999}) to general hypergraphs. 
  Let $G$ denote the graph obtained from $H = (V,E)$ by replacing each hyperedge $e = (v_1, \ldots, v_k) \in E$, by a clique with $\binom{k}{2}$ edges $(v_i,v_j)$, $i,j\in[k]$. We say that $G$ is associated to $H$. By construction, the degree of any vertex in $G$ is at most $d(r-1)$. 

  Graph $G$ will be important due to the following fact: it is enough to prove the inequality for those subsets $S\subseteq V$ that induce a connected subgraph of $G$. To see this, let $G[S]$ denote the subgraph induced by $S$.  Suppose $G[S]$ is not connected and let $S_1, \ldots, S_k \subseteq S$ be the vertex sets of the connected subgraphs.   If the lemma holds for connected components then $\left|2\cdot e_F(S_i) - e_E(S_i) \right| \leq 10 \sqrt{d \log (dr)}\cdot |S_i|$ for $i = 1, \ldots, k$, and so
  \begin{align*}
    \left|2\cdot e_F(S) - e_E(S) \right|  & = \left|\sum_{i=1}^k \left( 2\cdot e_F(S_i) - e_E(S_i)\right)\right|
    \leq \sum_{i=1}^k \left| 2\cdot e_F(S_i) - e_E(S_i)\right|  \\
    &\leq 10 \sqrt{d \log (dr)}\cdot \sum_{i=1}^k|S_i|
    =  10 \sqrt{d\log(dr)}\cdot |S|\,,
  \end{align*}
  where the first equality holds because there are no edges in $E$ (and $F \subseteq E$) between the sets $S_1, \ldots, S_k$.

  It is thus sufficient to prove the inequality for those sets $S$ that induce a connected subgraph $G[S]$.  Suppose we select $F$ by including each edge $e\in E$ with probability $1/2$ independently of other edges. That is, in the notation of Theorem~\ref{thm:LLL_constructive}, we have that $\mathcal{P}$ consists of $|E|$ mutually independent variables $\{P_e\}_{e\in E}$, where $P_e$ indicates whether $e\in F$ and $\Pr[P_e] = 1/2$. Now for each $S$ such that $G[S]$ is connected, let $A_S$ be the ``bad'' event that $\left|2\cdot e_F(S) - e_E(S) \right| > 10 \sqrt{d \log (dr)}\cdot |S|$. Note that $e_F(S)$ is the sum of at most $d|S|$ independent variables, attaining values $0$ and $1$, and that the expected value of $e_F(S)$ equals $e_E(S)/2$. Thus by the Chernoff inequality we get
  \begin{align*}
    \Pr[A_S]   < (dr)^{-6 |S|}\,.
  \end{align*}

  To apply Theorem~\ref{thm:LLL_constructive}, we analyze the dependency graph on the events: there is an edge between $A_S$ and $A_{S'}$ if $\vbl(A_S) \cap \vbl(A_{S'}) \neq \emptyset  \Leftrightarrow \delta_E(S) \cap \delta_E(S') \neq \emptyset$. Consider now a fixed event $A_S$ and let $k = |S|$. We  bound the number of neighbors, $A_{S'}$, of $A_S$ with $|S'| = \ell$.  Since we are interested in only subsets $S'$ such that $G[S']$ is connected, this is bounded by the number of distinct subtrees on $\ell$ vertices in the associated graph $G$, with a root in one of the endpoints of an edge in $\delta(S)$ . As $G$ has degree at most $d(r-1)$, there are at most $|S| + d(r -1)|S|=drk$ choices of the root. The number of such trees is known to be at most  (see e.g.~\cite{Frieze1999})
\begin{align}
  \label{eq:nrcomponents}
  drk \cdot {dr  (\ell-1) \choose \ell-1} \leq  drk \cdot (edr)^{\ell-1}\,,
\end{align}
where we used  that ${dr  (\ell-1) \choose \ell-1} \leq (edr)^{\ell-1}$.

Now to verify condition~\eqref{eq:cond} of Theorem~\ref{thm:LLL_constructive}, we  set $x(A_S) = (dr)^{-3|S|}$ for every bad event $A_S$. So if we consider an event $A_S$ with  $k = |S|$, then 
\begin{align*}
  x(A_S) \prod_{S': A_S \sim A_{S'}} \left( 1- x(A_{S'}) \right) &= (dr)^{-3k} \prod_{\ell=1}^n \left(1- (dr)^{-3\ell} \right)^{dr k (edr)^{\ell-1}}  \\
  & \geq (dr)^{-3k} \exp(- 2drk \sum_{\ell=1}^n (dr)^{-3\ell} (edr)^{\ell-1}) \\
  & \geq (dr)^{-3k} e^{-3k}  > (dr)^{-6k}/2 > \Pr[A_S]/2\,,
\end{align*}
where we used that $d$ is a sufficiently large constant, which is without loss of generality since if $d\leq 10 \sqrt{d\log (dr)}$ then the lemma becomes trivial.
In other words,~\eqref{eq:cond} is satisfied with $\epsilon$ set to $1/2$.

It remains to define an efficiently verifiable core subset $\mathcal{A}'\subseteq \mathcal{A}$ such that $1- \sum_{A \in \mathcal{A} \setminus \mathcal{A}'} x(A) \geq 1-n^{-3}$. We let
\begin{align*}
  \mathcal{A'} = \{A_S\in \mathcal{A}: |S| \leq s\} \mbox{ where $s= \log_{dr}(n)$}.
\end{align*}
By the same arguments as in~\eqref{eq:nrcomponents},  there is at most $n \cdot {dr  (\ell-1) \choose \ell-1} \leq n (edr)^{\ell-1}$ many events with $|S| = \ell$ (corresponding to connected components in $G$). Therefore, the following properties hold:
\begin{enumerate}
  \item $\mathcal{A}'$ is efficiently verifiable since it contains $n \cdot \sum_{\ell=1}^s (edr)^{\ell-1} = O(n \cdot (edr)^s)  = O(n^3)$ many events that can be efficiently enumerated by first selecting a vertex $r$ among $n$ choices and the considering all possible  trees rooted at $r$ with $\ell\leq s$ vertices.
  \item We have
    \begin{align*}
      \sum_{A_S \in \mathcal{A}\setminus \mathcal{A'}} x(A_S) & \leq \sum_{\ell = s+1}^n (dr)^{-6\ell} \cdot(n \cdot (edr)^\ell) \\
      & \leq n \cdot \sum_{\ell = s+1}^n (dr)^{-4\ell}  \leq n (dr)^{-4s}  = n^{-3}\,,
    \end{align*}
    where for the first inequality we again used that $d$ is a sufficiently large constant.
\end{enumerate}
We have verified Condition~\eqref{eq:cond} of Theorem~\ref{thm:LLL_constructive} and we have defined an efficiently verifiable core subset $\mathcal{A}'$ such that  $\sum_{A_S \in \mathcal{A}\setminus \mathcal{A}'} x(A_S) \leq n^{-3}$ and so the lemma follows.
\end{proof}

Applying the above lemma iteratively will  give us additive cut sparsifiers of constant degree. In particular, the condition in the following lemma will imply that the degree of each vertex in $\tilde{H}$ is at most $O(d_{\max}/2^k)$ and $k$ can be chosen so that the degree is at most $O(\frac{1}{\eps^2} \log(r/\eps))$. The following lemma therefore implies Theorem~\ref{th.probmax}.
\begin{lemma}
  There is an absolute constant $c$ such that the following holds. There is a probabilistic polynomial-time algorithm that given as input an $n$-vertex hypergraph $H =(V,E)$ of maximal degree $d_{\max}$ and of rank $r$, $\eps >0$, and any $k \in \mathbb{N}$ such that $d_{\max}2^{-k} \geq c \frac{1}{\eps^2}\log(r/\eps)$,  outputs a subgraph $\tilde{H} = (V, F)$ such that the following holds with probability at  least $1-n^{-2}$: 
  \begin{align*}
    \left|2^k\cdot e_F(S) - e_E(S) \right| \leq \eps d_{\max}|S| \qquad \mbox{ for every $S\subseteq V$.}
  \end{align*}
  \label{lemma:cutsparsifier}
\end{lemma}
\begin{proof}
  Starting with $H$ we apply Lemma~\ref{lemma:cut_iteration} $k$ times to obtain $\tilde{H}$. Let $F_i$ denote the edge set and let $d_i$ denote the maximum degree after round $i$. So $F_0 = E$ and $d_0=d_{\max}$.
  By the guarantees of Lemma~\ref{lemma:cut_iteration}, we have that with probability $1-n^{-3}$
\begin{align} 
  \label{eq:degree}
  |2 d_{i+1} - d_i | &\leq 10 \sqrt{d_i \log (d_ir)} 
\end{align}
and
\begin{align}
  \label{eq:set}
  \left|2\cdot e_{F_{i+1}}(S) - e_{F_i}(S) \right| &\leq 10 \sqrt{d_i \log (d_ir)}\cdot |S| \qquad \mbox{ for every $S\subseteq V$.}
\end{align}

As we apply Lemma~\ref{lemma:cut_iteration} $k$ times with $k \leq \log(n)$, the union bound implies that the above inequalities are true for all invocations of that lemma with probability at least $1 - k\cdot n^{-3} \geq 1-n^{-2}$. From now on we assume that the above inequalities hold and show that the conclusion of the statement is always true in that case.
Specifically, we now prove by induction on $k$ that 
\begin{align*} 
  |2^k d_{k} - d_0 | &\leq  \eps d_0\mbox{, and} \\
  \left|2^k\cdot e_{F_{k}}(S) - e_{F_0}(S) \right| &\leq \eps d_0 \cdot |S| \qquad \mbox{ for every $S\subseteq V$.}
\end{align*}
The claim holds trivially for $k=0$.
Assume it holds for all $i<k$, which in particular implies $2^id_i \leq 2 d_0$ for all $i <k$. By the triangle inequality and \eqref{eq:degree}, 
\begin{eqnarray*}
|2^{k} d_k - d_0  |  & \leq &    \sum_{i=0}^{k-1}  | 2^{i} (2d_{i+1} - d_{i})| \leq 10  \sum_{i=0}^{k-1}  2^{i} \sqrt{d_i \log (d_i r)}  \\
& \leq &  
10 \sum_{i=0}^{k-1}  2^{i}   \sqrt{2 (d_0/2^i) \log (2(d_0/2^i) r)} \qquad \text{(induction hypothesis on $d_i$).}
\end{eqnarray*}
As the terms increase geometrically in $i$, this sum is $O( 2^k\sqrt{(d_0/2^k) \log ((d_0/2^{k})r)}$ which is $\eps d_0$ by our assumption on $k$ and selection of $c$.

Finally, we note that $\left|2^k\cdot e_{F_{k}}(S) - e_{F_0}(S) \right| \leq \eps d_0 \cdot |S|$  follows by the same calculations (using~\eqref{eq:set} instead of~\eqref{eq:degree}). 
\end{proof}

\subsection{Additive spectral graph sparsifiers}
\label{sec:additive_spectral}
In this section we describe how the proof in the previous section generalizes to \emph{spectral} additive graph sparsifiers. 

\begin{theorem} \label{th.prob.graph.max}
Given an $n$-vertex graph $G=(V,E)$ and a parameter $\epsilon>0$, in probabilistic polynomial time we can find a subset $F\subseteq E$ of size $|F| = n \cdot  O((\log(1/\epsilon)^3/\epsilon^2)$ such that, if we let $L_G = D_G - A_G$
be the Laplacian of $G$, $L_{\tilde G}= D_{\tilde G} - A_{\tilde G}$ be the Laplacian of
the graph $\tilde G = (V,F)$,  and $c$ a normalization constant, we have
\begin{equation}
  - \epsilon d_{\max} I \preceq    c L_{\tilde G} - L_G \preceq  \epsilon d_{\max} I\,.
\end{equation}
\end{theorem}
\noindent Similar to before, this implies Theorem~\ref{th.prob.graph} by the reductions in Section~\ref{sec.reduction}.

To prove Theorem~\ref{th.prob.graph.max}, we need the following modification of  Lemma~\ref{lemma:cut_iteration} in the case of simple graphs. 
\begin{restatable}{lemma}{cutiterationadv}
  There exists a probabilistic polynomial-time algorithm that, given an $n$-vertex graph $G = (V,E)$ of maximal degree $d$, outputs a subgraph $\tilde{G} = (V,F)$ with $F\subseteq E$ such that the following properties  hold with probability at least $1-n^{-3}$: 
  \begin{enumerate}
    \item For  every disjoint $S,T \subseteq V$ we have 
      $\left|2\cdot e_F(S,T) - e_E(S,T) \right| \leq 10 \sqrt{d \log d}\cdot \sqrt{|S||T|}$. \label{en:adv_it1}
    \item For every vertex $v\in V$ we have $\left|2\cdot e_F(v) - e_E(v) \right| \leq 10 \sqrt{d\log d}$. \label{en:adv_it2}
  \end{enumerate}
  \label{lemma:cut_iteration_adv}
\end{restatable}
The above lemma is similar to Lemma 3.2 in~\cite{Bilu2006} with the exception that here we also need the degree constraints (the second condition). Similar to Lemma~\ref{lemma:cutsparsifier_adv} we obtain the following by applying Lemma~\ref{lemma:cut_iteration_adv} iteratively.
\begin{restatable}{lemma}{cutsparsifieradv}
  There is an absolute constant $c$ such that the following holds. There is a probabilistic polynomial-time algorithm that on input an $n$-vertex graph $G =(V,E)$ of maximum degree $d$, $\eps >0$, and any $k \in \mathbb{N}$ such that $d2^{-k} \geq c \frac{1}{\eps^2}\log(1/\eps)$,  outputs a subgraph $\tilde{G} = (V, F)$ such that the following properties hold with probability at  least $1-n^{-2}$: 
  \begin{enumerate}
    \item For  every disjoint $S,T \subseteq V$ we have 
      $\left|2^k\cdot e_F(S,T) - e_E(S,T) \right| \leq \eps d\cdot \sqrt{|S||T|}$. \label{en:adv1}
    \item For every vertex $v\in V$ we have $\left|2^k\cdot e_F(v) - e_E(v) \right| \leq \eps d$. \label{en:adv2}
  \end{enumerate}
  \label{lemma:cutsparsifier_adv}
\end{restatable}

The proofs of Lemma~\ref{lemma:cut_iteration_adv} and Lemma~\ref{lemma:cutsparsifier_adv} are very similar to the proofs of Lemma~\ref{lemma:cut_iteration} and Lemma~\ref{lemma:cutsparsifier}, respectively. We have therefore deferred them to Appendix~\ref{app:LLLproofs}. We now explain how Lemma~\ref{lemma:cutsparsifier_adv} implies an additive spectral sparsifier for graphs via the following result of Bilu and Linial~\cite{Bilu2006}:

\begin{lemma}[Lemma 3.3 in~\cite{Bilu2006}]
    Let $A$ be an $n\times n$ real symmetric matrix such that the $\ell_1$ norm of each row in $A$ is at most $\ell$, and all diagonal entries of $A$ are, in absolute value, $O(\alpha \log(\ell/\alpha) + 1))$. Assume that for any two vectors, $\bu,\bv\in \{0,1\}^n$, with $supp(\bu)\cap supp(\bv) = \emptyset$:
    \begin{align*}
        \frac{|\bu^T A \bv|}{\|\bu\| \|\bv\|} \leq \alpha\,.
    \end{align*}
    Then the spectral radius of $A$ is $O(\alpha(\log(\ell/\alpha) + 1))$.
    \label{lem:bilu}
\end{lemma}
Here $supp(\bu) = \{i: \bu_i \neq 0\}$ denotes the support of a vector $\bu$. Now let $G$ and $\tilde{G}$ be the input and output graph of Lemma~\ref{lemma:cutsparsifier_adv}. We set $A = 2^k L_{\tilde{G}} - L_{G}$. Since the Laplacian of a graph is a symmetric real matrix we have that $A$ is a symmetric $n\times n$ real matix where $n$ is the number of vertices in $G$ and $\tilde{G}$. We now verify that $A$ satisfies the assumptions of the above lemma assuming that the algorithm of Lemma~\ref{lemma:cutsparsifier_adv} was successful (which happens with probability at least $1-n^{-2}$).
\begin{itemize}
  \item The $\ell_1$ norm of a row in $A$ is at most the $\ell_1$ of that row in $2^k L_{\tilde{G}}$ plus the $\ell_1$ norm of that row in $L_{G}$. This can be upper bounded as follows. The $\ell_1$ norm of a row of a Laplacian matrix corresponding to a vertex $v$ equals twice the (weighted) degree of $v$. As any vertex in $G$ has degree at most $d$, it follows that the $\ell_1$ norm of any row in $L_G$ is at most $2d$. For  a row in $2^kL_{\tilde{G}}$ we use Property~\ref{en:adv2} of Lemma~\ref{lemma:cutsparsifier_adv} to bound the $\ell_1$ norm by $2(e_E(v) + 10 \sqrt{d \log d}) \leq 2(d + 10 \sqrt{d \log d})$. We therefore have that $\ell_1$ norm of any row in $A$ is bounded by
    \begin{align*}
      \ell = 2d + 2(d + 10\sqrt{d\log d}) = O(d)\,. 
    \end{align*}
  \item For the other two conditions, set $\alpha =  \epsilon d$ where $\epsilon$ is selected as in Lemma~\ref{lemma:cutsparsifier_adv}. Then we have that the absolute value of any diagonal entry in $A$ corresponding to a vertex $v$ equals
    \begin{align*}
      \left|2^k\cdot e_F(v) - e_E(v) \right| \leq \eps d\,,
    \end{align*}
    where the inequality is implied by Property~\ref{en:adv2} of Lemma~\ref{lemma:cutsparsifier_adv}. Similarly, consider any vectors $\bu,\bv \in \{0,1\}^n$ with $supp(\bu) \cap supp(\bv) = \emptyset$. Let $S = supp(\bu)$ and $T = supp(\bv)$. Then
    \begin{align*}
      \left| \bu^TA \bv \right| & = \left| \bu^T (2^k L_{\tilde{G}}) \bv - \bu^T L_G \bv \right| \\
      & = \| - 2^k\delta_F(S, T)  - (- \delta_E(S,T)) \| \\
      & \leq \eps d \cdot \sqrt{|S| |T|}\,,
    \end{align*}
    where the last inequality is implied by Property~\ref{en:adv1} of Lemma~\ref{lemma:cutsparsifier_adv}. The second equality is by the identity
    \begin{align*}
      \bu^T L_G \bv = \sum_{\{i,j\} \in E} \left(\bu_i \bv_i + \bu_j \bv_j - \bu_i \bv_j - \bu_j \bv_i\right) = - \delta(S, T)
    \end{align*}
    (and similar for $2^k L_{\tilde{G}}$).
\end{itemize}

We thus have that the assumptions of Lemma~\ref{lem:bilu} are satisfied with $\ell = O(d)$ and $\alpha = \eps d$. It follows that $A$ has a spectral radius of $O(\eps \log(1/\eps) d)$. Or equivalently:
\begin{align*}
    - c' \eps \log(1/\eps) d I \preceq 2^k L_{\tilde{G}} - L_{G} \preceq  c' \eps \log(1/\eps) d I\,,
\end{align*}
for an absolute constant $c'$.
To summarize, we obtain the following lemma which in turn implies Theorem~\ref{th.prob.graph.max} (by selecting $k$ as large as possible):
\begin{lemma}
  There are absolute constants $c$ and $c'$ such that the following holds. There is a probabilistic polynomial-time algorithm that on input an $n$-vertex graph $G =(V,E)$ of maximum degree $d$, $\eps >0$, and any $k \in \mathbb{N}$ such that $d2^{-k} \geq c \frac{1}{\eps^2}\log(1/\eps)$,  outputs a subgraph $\tilde{G} = (V, F)$ such that the following holds with probability at  least $1-n^{-2}$: 
  \begin{align*}
    - c' \eps \log(1/\eps) d I \preceq 2^k L_{\tilde{G}} - L_{G} \preceq  c' \eps \log(1/\eps) d I\,.
  \end{align*}
\end{lemma}

\section{Spectral Hypergraph Sparsification}
Let $H=(V,E)$ be a weighted hypergraph on $n$ vertices, with weights $w_e \geq 0$ on hyperedges $e \in E$. Let $r = \max_{e \in E} |e|$ be the maximum size of hyperedges in $H$, i.e., the rank of the hypergraph.

For a hyperedge $e$, the hypergraph Laplacian operator $Q_e:\R^n \rightarrow \R$, acts on a vector $x \in \R^n$ as
\[Q_e(x)  = w_e \max_{a,b \in e} (x_a-x_b)^2 = w_e \max_{a,b \in e} x^T L_{ab}x \]
where $L_{ab}$ is the standard graph Laplacian for an (unweighted) edge $ab$.
\begin{definition}(Hypergraph Laplacian)
Given a weighted hypergraph $H$, the hypergraph Laplacian operator $Q_H: \R^n \rightarrow \R$ for $H$ is defined as
\[ Q_H(x) = \sum_{e \in E(H)} Q_{e}(x)  = \sum_{e \in E(H)}  w_e \max_{a,b \in e }x^TL_{ab}x\]
\end{definition}

\begin{definition}(Multiplicative hypergraph spectral sparsifier.) A weighted hypergraph $\tilde{H}=(V,F)$  is a $(1+\eps)$-multiplicative spectral sparsifier of $H$ if
\beq
\label{eq:hyp-spars}
|Q_{\tilde{H}}(x) - Q_H(x)| \leq \eps Q_H(x)  \qquad \text{ for all $x \in \R^n$.}
\eeq
\end{definition}
We show the following result, which 
generalizes the result of Spielman and Srivastava \cite{SS11} from graphs to hypergraphs.
\begin{theorem}
For any hypergraph $H$ of rank $r$, and $\eps >0$, there is a   $(1+\eps)$-multiplicative spectral sparsifier $\tilde{H}$ of $H$ with $O(\frac{1}{\eps^2} r^3 n \log n)$ edges. 
Moreover, there is an efficient randomized algorithm that computes $\tilde{H}$ with probability $1-n^{-\Omega(1)}$, and runs in time $\tilde{O}_{r,\eps}(n)$.
\label{thm:hyp-spars-mult}
\end{theorem}
Unlike in the graph case, where it can be checked if $F$ satisfies \eqref{eq:hyp-spars} by an eigenvalue computation, we do not know of any efficient way to check Condition \eqref{eq:hyp-spars} for hypergraphs.

The following simple lemma shows that to prove Theorem \ref{thm:hyp-spars-mult}, it suffices to consider the case where all hyperedges have size between $r/2$ and $r$. We will make this assumption henceforth.
\begin{lemma}
\label{lem:almost-uniform}
If Theorem \ref{thm:hyp-spars-mult} holds for hypergraphs where each edge has size between $r/2$ and $r$, then it holds for all rank $r$ hypergraphs.
\end{lemma}
\begin{proof}
For $i=1,\ldots,\log r$, let $H_i$ be $H$ restricted to edges of size $(2^{i-1},2^i]$. 
For each $i$, we apply the claimed algorithm to $H_i$ to find a $(1+\epsilon_i)$-sparsifier  $\tilde{H}_i$ of $H_i$ with $\epsilon_i = \eps 2^{(i-\log r)/2}$ and return $\tilde{H} = \cup_i \tilde{H_i}$.

As $\tilde{H}_i$ has  $ 
O(\frac1{\eps_i^2} 2^{3i} n \log n) =  O(\frac{1}{\eps^2} 2^{\log r+2i}  n \log n)$ hyperedges, summing over all $i$ from $1$ to $\log r$ gives that $H$  has $O(\frac{1}{\eps^2} r^3 n \log n)$ hyperedges. Moreover, for any $x \in \R^n$,
$\tilde{H}$ satisfies \eqref{eq:hyp-spars} as
\begin{eqnarray*}   |Q_{\tilde{H}}(x) - Q_H(x)| & = &  | \sum_i( Q_{\tilde{H}_i}(x)  - \sum_i Q_{H_i}(x))|  
\leq \sum_i |Q_{\tilde{H}_i}(x) - Q_{H_i}(x)| \\
& \leq &  \sum_i \eps_i Q_{H_i}(x) \leq \eps \sum_i Q_{H_i(x)} = \eps Q_{H}(x)
\end{eqnarray*}
\end{proof}

\subsection{Algorithm}
The algorithm is a natural generalization of the sampling by effective resistances algorithm for graphs \cite{SS11}. 

\begin{definition}(Associated graph.)
Let $G$ denote the multi-graph obtained by replacing each hyperedge $e = (v_1,\ldots,v_k) \in H$, by a clique with $\binom{k}{2}$ edges $(v_i,v_j)$, $i,j\in[k]$, each with the same weight as that of $e$. 
We call $G$ the associated graph of $H$.
\end{definition}

 To avoid confusion, we will use $(a,b)$ to denote the edges in $G$ and $e$ for the hyperedges in $H$.

\paragraph{Algorithm.}
Given the hypergraph $H$, let $G$ be its associated graph, and let  $L_G = \sum_{(ab) \in E} L_{ab}$ be the (graph) Laplacian
of  $G$. 
Let $Y_{ab} = L_G^{-1/2} L_{ab} L_G^{-1/2}$, where $L_G^{-1}$ is the pseudoinverse of $L_G$.
Then 
$r_{ab} := \|Y_{ab}\|$
is the effective resistance of the edge $ab$. 
For a hypergraph $e \in E(H)$, define
\[ r_e = \max_{a,b \in e} r_{ab} \]

Let $L = c \eps^2/ (r^4 \log n)$, where $c$ is a fixed constant that can be computed explicitly from the analysis described later. For each hyperedge $e$, set 
\[ p_e = \min(1,\frac{r_e}{L}).\]

$\tilde{H}$ is obtained by sampling each $e \in H$ independently with probability $p_e$ and scaling its weight by $1/p_e$.

\subsection{Analysis}
Our goal in the next few sections is to prove Theorem \ref{thm:hyp-spars-mult}.
We first show that  $\tilde{H}$ has $O((r^3 n\log n)/\eps^2)$ edges with high probability, and then focus on showing that \eqref{eq:hyp-spars} holds with probability $1-n^{\Omega(1)}$.

\paragraph{Bounding the number of edges.}
The expected number of edges in $\tilde{H}$ is 
$\sum_e p_e$, which is at most $(\sum_e r_e)/L$. So it suffices to bound,
\[ \sum_{e \in E(H)} r_e = \sum_{e \in E(H)} \max_{a,b \in e} r_{ab}.\]
The effective resistances in a graph satisfy the metric property,
$r_{ab} \leq r_{ac}  + r_{cb}$ for all $a,b,c$, and so for any $e \in E(H)$ with $k=|e|$, and  any $a,b \in e$, summing over all $c \in e$ gives
\[ k r_{ab} \leq \sum_{c \in e} (r_{ac} + r_{cb}) \leq 2 \sum_{c,d \in e} r_{cd} \]
As $k \geq r/2$ by our assumption from Lemma \ref{lem:almost-uniform}, this gives that
\[ \sum_{e \in E(H)} r_e 
\leq  \sum_{e \in E(H)} \frac{4}{r} \sum_{a,b \in e} r_{ab}
=  \frac{4}{r} \sum_{(ab) \in E(G)} r_{ab}.\]
Without loss of generality we can assume that $G$ is connected, in which case $L_G$ has rank exactly $n-1$ and $L_G \mathbf{1}=0$. This gives that $\sum_{(ab) \in E(G)} Y_{ab} = I_{n-1}$, which upon taking traces on both sides, and using that $Y_{ab}$ is rank $1$, gives
$\sum_{(ab) \in E(G)} r_{ab} = n-1.$

So the expected number of edges is $O(n/rL) = O((nr^3 \log n)/\eps^2)$, and
as the hyperedges are sampled independently, by standard tail bounds the number of edges is tightly concentrated around the mean.

\paragraph{Proving condition \eqref{eq:hyp-spars}.}
We now focus on showing that \eqref{eq:hyp-spars} holds. It is useful to first consider the analysis of Spielman and Srivastava \cite{SS11}
for the graph case.

\paragraph{The graph case.}
In the graph setting, \eqref{eq:hyp-spars} becomes
\beq | x^T (L_{\tilde{G}}-L_G)x | \leq \eps x^T L_Gx  \qquad \text{for all $x \in \R^n$},
\label{eq:ss}
\eeq
where $L_{\tilde{G}}  = \sum_{(ab) \in F}  (1/p_{ab}) L_{ab}$ is the Laplacian of $\tilde{G}$.

Setting $z = L_G^{1/2} x$, and $Y_{ab} = L_G^{-1/2} L_{ab} L_G^{-1/2}$, this is equivalent to showing that 
\[ \sum_{(ab) \in G}   z^T (X_{ab} - Y_{ab}) z \leq \eps \|z\|^2 \qquad  
\text{ for all $z$ in the range of $L_G$.}
\]
where $X_{ab}$ is the random matrix which is $Y_{ab}/p_{ab}$ with probability $p_{ab}$ and is the  all-$0$ matrix otherwise. So $\E[X_{ab}]= Y_{ab}$.
As $\sum_{(ab) \in G} Y_{ab} = I$ (on the range of $L_G$), this reduces to show that $z^T (\sum_{(ab) \in G } X_{ab}  - I)  z \leq \eps \|z\|^2$ 
or equivalently,
\beq
\label{ss:eq2}
\|\sum_{ab} (X_{ab} - \E[X_{ab}]) \| \leq \eps. 
\eeq
This can be done using  standard matrix concentration bounds for the spectral norm such as the following.
\begin{theorem}(Matrix Bernstein inequality, \cite{tropp2015}.)
\label{th:mb2}
Let $X_1,\ldots,X_m$ be independent, symmetric  $d\times d$  random matrices, and $S = \sum_i X_i$,  $L = \max_i \|X_i\|$. Then
\[ \Pr [\|S - \E[S] \| \geq t] \leq d \exp\left ( - \frac{t^2/2}{ \| \sum_i \E[X_i^2] \| + Lt/3}\right) \]
\end{theorem}
In particular, this gives the following useful corollary. 
\begin{corollary}
\label{cor:mb} If
$A_1,\ldots,A_m$ are PSD with $\sum_i A_i \preceq I$, and $X_i = A_i/p_i$ with probability $p_i$ and $0$ otherwise, then
for any $\eps\leq 1$,
\[ \Pr[ \|S - \E[S] \| \geq \eps ]  \leq d \exp(-\eps^2/3L)\]
where $L = \max_i \|A_i\|/p_i$.
\end{corollary}
Applying Corollary \ref{cor:mb} with $A_i = Y_{ab}$ and  $p_i = p_{ab}$, we have
$L= \max_{ab} \|Y_{ab}\|/p_{ab} = \max_{ab} r_{ab}/p_{ab} =  O(\eps^2/\log n),
$
which gives that \eqref{ss:eq2} holds with probability at least $1-n^{-\Omega(1)}$ as desired.

\paragraph{The hypergraph case.} We first reduce the condition \eqref{eq:hyp-spars} for hypergraphs to a simpler form. 
Let $G$ be the graph associated to $H$ and $L_G$ be its Laplacian. We have following simple relation.

\begin{lemma}
\label{lem:hyp-simple}
For a $k$-edge $e$, let $L_e = \sum_{a,b  \in e} L_{ab}$. Then, for all $x \in \R^n$
\[  \frac{2}{k(k-1)} x^T L_e x  \leq Q_e(x) \leq \frac{2}{k} x^T L_e x \]
If the hyperedges in $H$ have size in $(r/2,r]$, then
for all $x \in \R^n$
\[  \frac{2}{r(r-1)} x^T L_G x  \leq Q_H(x) \leq  \frac4r x^T L_G x. \]
\end{lemma}
\begin{proof} 
Suppose that $x_1 \leq \ldots \leq x_k$.
Then $Q_e(x) = (x_k - x_1)^2$, while $e$ contributes $\sum_{i,j\in[k]} (x_i-x_j)^2 $ to $x^T L_Gx$.
So the lower bound in the  first inequality follows directly. 

For the upper bound, we observe that $(x_k - x_1)^2 \leq 2(x_k-x_j)^2 + 2(x_j-x_1)^2$ for each $j = 2,\ldots,k-1$.
Summing these gives 
$ (k-2) (x_k - x_1)^2 \leq 2 \sum_{j=2}^{k-1} ((x_k-x_j)^2 + (x_j-x_1)^2)$.
Adding $2 (x_k-x_1)^2$ to both sides, and noting that the resulting right side is at most $2 x^T L_e x$, the upper bound follows.

Summing up over all $e \in E(H)$, and using $r/2 < k \leq r$ gives the second set of inequalities.
\end{proof}

By Lemma \ref{lem:hyp-simple}, to show \eqref{eq:hyp-spars} it suffices to show that for $x \in \R^n$,
\beq |Q_{\tilde{H}}(x) - Q_H(x)| \leq \frac{\eps}{r^2} x^T L_G x
\label{eq:hyp-spars-1}
\eeq 
As before, setting $z = L_G^{1/2} x$ and $ Y_{ab} =L_G^{-1/2} L_{ab} L_G^{-1/2}$ gives 
\[ Q_e(x) = \max_{a,b\in e} x^T L_{ab} x = \max_{a,b\in e} z^T Y_{ab} z.\]
Let us define  \[ W_e(z) = \max_{a,b\in e} z^T Y_{ab} z,\]
and let $X_e$ be the random variable that is $1/p_e$ with probability $p_e$ and $0$ otherwise. Then \eqref{eq:hyp-spars-1} is equivalent to  
\beq |\sum_{e \in H}  (X_e -1) W_e(z) | \leq \frac{\eps}{r^2} \|z\|^2 \qquad \text{ for all } z \in \text{Im}(L_G).
\nonumber
\eeq 
As $W_e(z)$ scales as $\|z\|^2$, it suffices to show that
\beq |\sum_{e \in H}  (X_e -1) W_e(z) | \leq \frac{\eps}{r^2} \|z\|^2 \qquad \text{ for all } z \in B_2,
\label{eq:hyp-spars-2}
\eeq 
where $B_2$ is the unit $\ell_2$-ball in the subspace restricted to the image of $L_G$,

However, unlike in the graph case, it is not immediately clear how to show concentration to prove \eqref{eq:hyp-spars-2}. In particular, as the operator $W_e(z)$ involves the $\max$ term, the left hand side does not correspond to any standard linear-algebraic quantity  like the spectral norm, for which we can use matrix concentration bounds.

A natural idea might be to replace $W_e(z)$
by the larger term $\sum_{(a,b) \in e} z^T Y_{ab} z$, and reduce the problem to the graph case, for which we can use matrix Bernstein inequality. But this does not work as the multiplier $(X_e -1)$ in \eqref{eq:hyp-spars-2} can be negative (so $|\sum_{e \in H}  (X_e -1) W_e(z) |$ could be arbitrarily large even though $|\sum_{e \in H}  (X_e -1) \sum_{a,b \in e}  z^T Y_{ab} z |$ is $0$).

So our approach will be to directly consider the inequality \eqref{eq:hyp-spars} for each $z$ in the unit $\ell_2$-ball, and bound  the probability of violation for any $z$ by applying a union bound over all such points $z$ by a careful net argument. More precisely, we view the left hand side of $\eqref{eq:hyp-spars-2}$ as a random process indexed by $z \in B_2$, and use generic chaining arguments to bound the supremum of this process.


Summarizing, let
$W_H(z) = \sum_{e \in E} W_e(z)$, and
$W_{\tilde H}(z)=  \sum_{e \in E} X_e W_e(z)$ be the corresponding operator for $\tilde{H}$. 
Proving Theorem \ref{thm:hyp-spars-mult} reduces to the following.
\begin{theorem}
\label{thm:hyp-mult-red}
With probability $1-n^{-\Omega(1)}$, it holds that
    \begin{equation}
    \label{eq:hyp}
        |W_{\tilde H}(z) - W_H(z)| \leq \frac{\eps}{r^2}  \qquad {\text{for all $z \in B_2$}}
    \end{equation} 
\end{theorem}
This will be accomplished in the next few sections.

\subsection{Supremum of random processes}
We first give some background on the theory of supremum of random processes and mention the results we need. For more details, we refer the reader to Chapters 7 and 8 of the excellent recent text \cite{ver-book}.
\begin{definition}(Random process)
A {\em random process} is a collection of random variables $(X_t)_{t \in T}$ on the same probability space, which are indexed by the elements $t$ of some set $T$.
\end{definition}
The random variables $X_t - X_s$ for $s,t \in T$ are the increments of the random process. A random process is called mean-zero if all $X_t$ have mean-zero. We will only consider mean-zero  processes in this paper.
\begin{definition}(Gaussian process) 
A random process $(X_t)_{t \in T}$ is called a Gaussian process, if $X_t$ are jointly Gaussian, i.e.~if every finite linear combination of the $X_t$ is Gaussian. 
\end{definition}
Any Gaussian process can be written in a canonical way as 
$X_t = \ip{g,t}$, where $t \in \R^n$ and $g \sim N(0,I_n)$ is the standard normal vector.
This gives that  for any $s,t \in T$, the increments of a Gaussian process satisfy, \[
(\E [(X_t-X_s)^2])^{1/2} =
\|t-s\|^2 \]
where $\|t-s\|_2$ denotes the Euclidean distance between $t$ and $s$. 

As a (mean-zero) Gaussian process is completely determined by its covariance, the supremum
$\E \sup_{t \in T} X_t$ of a gaussian process is completely determined by the geometry of the metric space $(T,d)$. In particular, we have the following celebrated result.
\begin{theorem}(Talagrand's majorizing measures theorem.))
Let $(X_t)_{t \in T}$ be a mean-zero Gaussian process on a set $T$, with the canonical metric on $T$,  $d(s,t) = \|t-s\|_2$.
Then for some absolute constants $c,C$
\[ c \gamma_2(T,d) \leq \E \sup_{t\in T} X_t  \leq C \gamma_2 (T,d) \]
where  $\gamma_2(T,d) = \inf_{(T_k)} \sup_{t \in T} \sum_{k=0}^\infty 2^{k/2} d(t,T_k),$ and
where the infimum is over all sets $T_k \subset T$, satisfying $|T_k| \leq 2^{2^k}$ for all $k$.
\label{th:tal-maj}
 \end{theorem}


We now consider sub-gaussian processes (see section 8.1 in \cite{ver-book} for details). 
\begin{definition}(Sub-gaussian increments.)
Consider a random process $(X_t)_{t \in T}$ on a metric space $(T,d)$. We say that the process has sub-gaussian increments if there exists some $K \geq 0$, such that 
\[ \|X_t -X_s\|_{\psi_2} \leq K d(t,s) \qquad \text{for all $t,s \in T$.}\]
\end{definition}
Here $\|\cdot \|_{\psi_2}$ is the {\em sub-gaussian norm} for real-valued random variable $X$, defined as 
\[ \|X\|_{\psi_2} = \inf \{ t>0: \E [\exp(X^2/t^2)] \leq 2 \}. \]
We need the following two basic facts about the $\psi_2$-norm
(section 2.6 in \cite{ver-book}).
\begin{fact}
\label{th:psi-max}
For any random variable $X$, $\|X\|_{\psi_2} \leq c \|X\|_\infty$ (with $c = 1/\sqrt{\ln 2}$).
\end{fact}
\begin{fact}
For $X_1,\ldots,X_n$ independent
$ \|\sum_{i=1}^n X_i\|_{\psi_2}^2 \leq c \sum_{i=1}^n \|X_i\|^2_{\psi_2},$
where $c$ is an absolute constant.
\label{th:psi2-sum}
\end{fact}

The following result follows directly from Theorem \ref{th:tal-maj} (see section 8.6 in \cite{ver-book} for details).
\begin{theorem}
\label{tal:comp}
(Talagrand's comparison inequality.) Let 
$(X_t)_{t\in T} $ be a mean-zero random process on a set $T$, and let $(Y_t)_{t \in T}$ be a Gaussian process with the canonical metric $d(s,t) = (\E[(Y_s-Y_t)^2])^{1/2} = \|Y_s-Y_t\|_2$. Assume that for all  $s,t \in T$, we have 
\[ \|X_t-X_s\|_{\psi_2} \leq K \|Y_t-Y_s\|_2. \]
Then, for some absolute constant $C$,
\[\E \sup_{t \in T} X_t \leq C K \E \sup_{t \in T} Y_t\]
More generally, for every $u\geq 0$, 
\[
\Pr \left[ \sup_{s,t \in T} X_t-X_s   \geq  C K \left(\E \sup_{t \in T} Y_t +  u\ \mathrm{diam}(T) \right) \right]   \leq 2 \exp(-u^2),\] where $\mathrm{diam}(T)$ is the diameter of $T$ with respect to the metric $d$.
\end{theorem}

In other words, if we can find a Gaussian process $Y_t$ such that its Gaussian increments upper bound the corresponding sub-gaussian increments of $X_t$, then we can bound the supremum of $X_t$ by that of $Y_t$.

\subsection{Random process for hypergraph sparsification}
We now consider the relevant 
random processes arising in our setting of hypergraph sparsification.

\paragraph{Gaussian Process on the associated graph.}
Let $G$ be the associated graph of $H$, and consider the random matrix $U = \sum_{(ab) \in E(G) } g_{ab} Y_{ab}$,
where $g_{ab}$ are independent $N(0,1)$. 

For $z \in \R^n$, consider the  Gassian process $U_z = z^T U z = \sum_{(ab) \in E(G)} (z^TY_{ab} z) g_{ab}$.
As $\|U\| = \max_{ z \in B_2} z^T Uz$, it follows that  $\|U\| = \sup_{z \in T} U_z$ with $T=B_2$. As,
 \[ U_{z} - U_{z'} = \sum_{ab} (z^T Y_{ab} z - z'^T Y_{ab} z') g_{ab}, \]
 the canonical metric induces the distance
 \beq
 \label{proc-gaus}
 d_u(z,z')^2 :=  \E[ (U_z - U_{z'})^2]  = 
\sum_{ab} (z^T Y_{ab} z - {z'}^T Y_{ab} z')^2.
\eeq
\paragraph{Hypergraph sampling process.} Let us now consider the random process corresponding to  \eqref{eq:hyp-spars-2}.
We consider the case when $p_e = 1/2$ (the theory of sub-gaussian does not work well for $p_e \ll 1$) (in section \ref{s:gen-case} we will show that the case of general $p_e$ reduces to that of $p_e=1/2$.
For $p_e=1/2$, $(X_e-1)$ takes value $-1$ or $1$  with probability $1/2$ each. So we define 
\[ V_z :=  \sum_{e \in E(H)}   \eps_e  W_e(z) = \sum_{e \in E(H)} \eps_e \max_{a,b \in e} z^T Y_{ab}z \]
where $\eps_e$ are independent Rademacher random variables.

The following key Lemma will allow us to bound the (complicated) sub-gaussian process $V_z$ by the simpler Gaussian process $U_z$.
\begin{lemma}
\label{lem:connect}
There is an absolute constant $c$, such that for any $z,z' \in B_2$, 
\[ \|V_z - V_{z'}\|_{\psi_2} \leq c  \|U_z-U_{z'}\|_2.\]
\end{lemma}
Before proving this lemma, we need the following simple fact.
\begin{lemma}
\label{lem:basic-ineq}
For any numbers $c_1,\ldots,c_s$ and $d_1,\ldots,d_s$, 
\[ (\max_i c_i - \max_i d_i)^2 \leq  \sum_i (c_i -d_i)^2 \]
\end{lemma}
\begin{proof}
Let $c_a = \max_i c_i$ and $d_b = \max_i d_i$.
If $c_a \geq d_b$, then 
\[ |c_a - d_a| = c_a - d_a \geq c_a - d_b =\max_i c_i - \max_i d_i \geq  | \max_i c_i - \max_i d_i|. \]
The other case when $c_a \leq d_b$ is completely analogous.
\end{proof}
We now prove Lemma \ref{lem:connect}.
\begin{proof}(Lemma \ref{lem:connect}).
Fix $z, z' \in B_2$. 
For a hyperedge $e$,
let $a(e),b(e) \in e$ be the indices that maximize $z^T Y_{a(e)b(e)} z$, and $a'(e),b'(e) \in e$ be those that maximize $z'^T Y_{a'(e)b'(e)} z'$. Then,
\[ V_z - V_{z'} =  \sum_{e \in E(H)}  \eps_e (z^T Y_{a(e)b(e)} z  - z'^T Y_{a'(e)b'(e)}z').\]
By Facts \ref{th:psi-max} and \ref{th:psi2-sum}, there is an absolute constant $c$ such that,
\beq 
\label{psi2-vz}
\|V_z - V_{z'}\|^2_{\psi_2}  \leq c 
 \sum_{e \in E(H)}  \left(z^T Y_{a(e)b(e)} z - z'^T Y_{a'(e)b'(e)}z' \right)^2 
\eeq
On the other hand, by \eqref{proc-gaus} we have that
\beq 
\label{eq:guass-bd} 
\| U_z - U_{z'}\|^2_2 := d_u(z,z')^2 =  \sum_{(ab) \in E(G)}   (z^T Y_{ab} z -  z'^T Y_{ab}z')^2 
\eeq
Even though $Y_{a(e)b(e)}$ could be different from $Y_{a'(e)b'(e)}$, we can use Lemma \ref{lem:basic-ineq} to show that 
 the right hand side of \eqref{psi2-vz} is upper bounded by the right side of \eqref{eq:guass-bd}.


Fix a hyperedge $e \in H$ and let $k=|e|$. Applying Lemma \ref{lem:basic-ineq} 
to the $s=k^2$ pairs $a,b\in [k]$ with $c_{ab}=z^T Y_{ab}z$ and $d_{ab} = z'^T Y_{ab} z'$, we get
\beq
\label{eq:key-ineq}
 \left(z^T Y_{a(e)b(e)} z - z'^T Y_{a'(e)b'(e)}z' \right)^2 
 \leq  \sum_{a,b \in e}  \left(z^T Y_{ab} z - z'^T Y_{ab}z' \right)^2. 
\eeq
Summing over the hyperedges $e \in E(H)$, using \eqref{psi2-vz} and \eqref{eq:guass-bd}, and noting that 
\[ \sum_{e \in E(H)} \sum_{a,b \in e}  \left(z^T Y_{ab} z - z'^T Y_{ab}z' \right)^2 =  \sum_{(ab) \in E(G)}  \left(z^T Y_{ab} z - z'^T Y_{ab}z' \right)^2 
\]
 gives the result.
\end{proof}

\noindent {\bf Remark:} 
It might seem that the inequality \eqref{eq:key-ineq} can be tightened by a factor $O(1/r)$, by using that for any $z_1 \leq \ldots \leq  z_r$, we have that 
$(z_r - z_1)^2 \leq  \frac{2}{r} \sum_{ij} (z_i-z_j)^2$ (we used similar ideas in Lemma \ref{lem:hyp-simple}).
However, this following example shows that this is not possible.

Suppose, $Y_{ab} = L_{ab}$, i.e.~$z^T Y_{ab} z = (z_a-z_b)^2$.
Consider $z = (z_1,\ldots,z_r)= (-1,0,\ldots,0,M)$ and $z'=(0,0,\ldots,0,M)$.
The term on the left side of \eqref{eq:key-ineq} is 
\[ (\max_{ab \in e} z^T Y_{ab} z  - \max_{a'b' \in e} z'^T Y_{a'(e)b'(e)} z')^2  = ((M+1)^2  -   M^2)^2  \approx 4M^2. \]
On the other hand, the terms on the right side of \eqref{eq:key-ineq}
correspond to 
\[   (z^TY_{ab}z - z'^TY_{ab} z')^2  =  ((z_a-z_b)^2 - (z'_a-z'_b)^2)^2.\]
However, it is easily verified that for each of the $r^2-1$ pairs $(a,b)$ except $(a,b)= (1,r)$,  
$(z_a-z_b)^2 - (z'_a -z'_b)^2 \leq 1$.
Making $M$ arbitrarily large, this shows that we really need the full contribution of each of the $r^2$ terms on the right side of \eqref{eq:key-ineq}, and we cannot
improve the inequality.

\paragraph{Bounds on the process $V_z$.}
Lemma \ref{lem:connect} and Theorem \ref{tal:comp} will let us bound the supremum of $V_z$ by that of $U_z$. We can directly bound the latter using the following variant of the matrix Bernstein inequality.
\begin{theorem}(\cite{tropp2015}, Theorem 4.1.1.)
\label{th:mb1}
If $A_1,\ldots,A_m$ are symmetric $d \times d$ matrices, and $g_i$ are independent $N(0,1)$ random variables, then for $Y = \sum_i g_i A_i$
\[  \Pr[ \|Y\| \geq t ] \leq d \exp(-t^2/2\| \sum_i A_i^2 \|)\]
\end{theorem}
In particular, Theorem \ref{th:mb1} has the following corollary.
\begin{corollary}
\label{cor:mb2}
If $A_i$ are PSD, and satisfy  $\|A_i \| \leq \delta$ and $\sum_i A_i \preceq I$, then
$\| \sum_i A_i^2 \| \leq \delta$ and so for $c\geq 2$
\[ \Pr[ \| Y \|   \geq  c \sqrt{\delta \log d}] \leq d \exp(-(c^2 \log d) /2) \leq   d^{-c^2/4}.\]
\end{corollary}
We now show that a similar  tail bound holds for $\sup_z V_z$. 
\begin{theorem}
\label{thm:subset-process}
Let $S \subset E(H)$ be a subset of hyperedges with $r_e  \leq \delta$ for all $e \in S$. For independent Rademacher $\eps_e$, and $z \in \R^n$, let 
\[V_z = \sum_{e \in S} \eps_e W_e(z).\]
Then  $\E \sup_{z \in B_2} V_z =  O(\sqrt{\delta \log n})$, and for all $u\geq 0$
 \[ \Pr \left[\sup_{z \in B_2} V_z \geq O(\sqrt{\delta \log n} +  2u \sqrt{\delta})   \right] \leq 2 \exp(-u^2).\]
\end{theorem}
\begin{proof}
Let $E(G[S]) = \{(ab): a,b \in e, e \in S\}$ be the multi-set of edges in the associated graph $G[S]$ for $H$ restricted to $S$.
Consider the process $U_z = \sum_{(ab) \in E(G[S])}  g_{ab} z^T Y_{ab} z$ where $g_{ab}$ are independent $N(0,1)$. Then 
\[ \sup_{z \in B_2} U_z = \|U\|, \quad \text{ where } U = \sum_{(ab) \in E[G[S])} g_{ab} Y_{ab}\]
As  $r_e= \max_{a,b\in e} \|Y_{ab}\|$, we have $\|Y_{ab}\| \leq \delta$ for all $(ab) \in E(G[S])$. Moreover as 
\[ \sum_{(ab) \in E(G[S])} Y_{ab} \leq \sum_{(ab) \in E(G) } Y_{ab} \preceq I,\]
Corollary \ref{cor:mb2} gives that $\|U\| = O(\sqrt{ \delta \log n})$.
By Lemma \ref{lem:connect} and the Talagrand comparison inequality Theorem \ref{tal:comp}, we have that for some constant $C$,
\[\sup_{z \in B_2} V_z \leq C \sup_{z \in B_2} U_z = C \|U\| = O(\sqrt{ \delta \log n}). \] 
To compute the tail bound on $\sup_z V_z$, we need to compute the  diameter $\mathrm{diam}(T)$ with respect $d_u$.
By the definiton of $d_u(z,z')$,
\[ \mathrm{diam}(T)^2 = \max_{z,z'} d_u(z,z')^2 = 
\sum_{(ab) \in E(G[S])} (z^T Y_{ab} z - z'Y_{ab}z')^2 \]
Using $(c-d)^2 \leq 2c^2 + 2d^2$ for any $c,d \in \R$,
\begin{eqnarray*}
\mathrm{diam}(T)^2 &\leq &  4  \max_{z \in T}  \sum_{ab \in E(G[S])}  (z^T Y_{ab} z)^2 
\\  
& \leq &  4  \max_{z \in T} \big( (\max_{ab \in E(G[S])} z^T Y_{ab} z)  \   (\sum_{ab \in E(G[S])} z^T Y_{ab} z) \big) \\
& \leq &   4 \delta  \qquad \qquad (\text{ as } T = B_2, \|Y_{ab}\| \leq \delta, \sum_{ab} Y_{ab} \preceq I) 
\end{eqnarray*}
The claimed tail bound on $\sup_z V_z$ now follows from Theorem \ref{tal:comp}.
\end{proof}

\subsection{Putting it all together}
\label{s:gen-case}
We now  prove Theorem \ref{thm:hyp-mult-red}.
Given $H$, we compute $G$ and $p_e$ as described earlier.
By rounding $p_e$ up to nearest integer powers of $2$, we can assume that for each $e \in E(H)$, $p_e = 2^{-j}$ for some $j\in\{0,\ldots,\ell\}$. This ensures $p_e \geq  \min(1,r_e/L)$, while at most doubling the expected size of $\tilde H$. 
Let $C_j = \{e \in E(H): p_e=2^{-j}\}$.
As $p_e=1$ for hyperedges $e$ with $r_e \geq L$, the sampling error in $\tilde{H}$ is only due to edges with $r_e < L$, and so in the analysis of the sampling error below we will assume that $r_e < L$ for all $e \in E(H)$.

We view the process of sampling $\tilde H$ in the following iterative way. 
Let $H_{0}=H$, and for $i=1,\ldots,\ell$,
$H_i$ is obtained from $H_{i-1}$ by picking each hyperedge $e$ of classes $C_j$  for $j \in \{\ell-i+1, \ell\}$, independently with probability $1/2$, and doubling the weight of $e$ if it is picked.
Or equivalently, for $i=1,\ldots, \ell$, $H_i$ is obtained by picking each edge $e \in C_j$ in $H$ independently with probability $\min(1,2^{\ell-j-i})$
and scaling its weight by $\max(1,2^{j+i-\ell})$.
So $H_{\ell}=\tilde H$, and an edge $e$ in $C_j$ survives independently in $H_{\ell}$ with probability $p_e =2^{-j}$.

\begin{proof}(Theorem \ref{thm:hyp-mult-red}.)
By the discussion above, for $i=0,\ldots,\ell$, 
\[W_{H_i}(z) = \sum_{j=0}^\ell \sum_{e \in C_j \cap E(H_i)}  \max(1,2^{i+j-\ell}) W_e(z).\] 
 and note that $W_{H_0}(z) = W_H(z)$ and $W_{\tilde H}(z)= W_{H_\ell}(z)$.

For any $z \in B_2$, by triangle inequality 
\[ |W_{\tilde H}(z) - W_H(z) | = |W_{H_\ell}(z) - W_H(z) | \leq \sum_{i=1}^{\ell} |W_{H_i}(z) - W_{H_{i-1}}(z) |\]
Taking supremum over all $z$, and taking the $\sup$ inside the summation,
\beq  \sup_z |W_{\tilde H}(z) - W_H(z) |  \leq \sum_{i=1}^{\ell} \sup_z |W_{H_i}(z) - W_{H_{i-1}}(z) |
\label{eq:sup-delta}
\eeq

As $H_i$ is obtained by $H_{i-1}$ by sampling each edge of class $j \in [\ell-i+1,\ell]$ with probability $1/2$ and doubling its weight, we have
\[W_{H_i}(z) - W_{H_{i-1}}(z)  = \sum_{j=\ell-i+1}^\ell \sum_{e \in C_j \cap E(H_{i-1})} \eps_e 2^{i+j-\ell-1} W_e(z)\]

For $j \in \{\ell-i+1,\ell\}$ and any $e \in C_j$,  $W_e(z) = \max_{a,b \in  e} z^T Y_{ab} z$ with
$\|Y_{ab}\| \leq r_e \leq 2^{-j}L$ for all $a,b \in e$.
So $\| 2^{i+j-\ell} Y_{ab} \| \leq 2^{i - \ell}L $, and applying Theorem \ref{thm:subset-process} with $V_z =W_{H_i}(z) - W_{H_{i-1}}(z)$ and $u  = \sqrt{\log n}$  gives that 
\[ \Pr [\sup_z V_z  \geq O(\sqrt{2^{i-\ell}L \log n})] \leq n^{-\Omega(1)} \]
Together with \eqref{eq:sup-delta}, and taking union bound over the $\ell = O(\log n)$ classes, we get that
 \[ \sup_z |W_{\tilde H}(z) - W_H(z) |  \leq O(\sum_{i=1}^\ell \sqrt{2^{i-\ell}L \log n}) = O\left(\frac{\eps}{r^2}\right),\]
 with probability $n^{-\Omega(1)}$, as desired. 
\end{proof}

\appendix
\section{Proofs of Lemma~\ref{lemma:cut_iteration_adv} and Lemma~\ref{lemma:cutsparsifier_adv}}
\label{app:LLLproofs}

For completeness we give the proofs of Lemma~\ref{lemma:cut_iteration_adv} and Lemma~\ref{lemma:cutsparsifier_adv} that are very similar to those of Lemma~\ref{lemma:cut_iteration} and Lemma~\ref{lemma:cutsparsifier}, respectively.

\cutiterationadv*
\begin{proof}
  The proof closely follows that of  Lemma~\ref{lemma:cut_iteration} and is similar to the proof of Lemma 3.2 by Bilu and Linial~\cite{Bilu2006}.
  We first observe that we only need to verify Property~\ref{en:adv_it1} for those disjoint sets $S, T \subseteq V$ such that $G[S\cup T]$ is connected. To see this, let $G[S\cup T]$ denote the subgraph induced by $S\cup T$.  Suppose $G[S\cup T]$ is not connected and let $S_1 \cup T_1, S_2 \cup T_2, \ldots, S_k \cup T_k$ be the vertex sets of the connected components where $S_1, \ldots, S_k \subseteq S$ and $T_1, \ldots, T_k \subseteq T$.   If Property~\ref{en:adv_it1} holds for connected components then $\left|2\cdot e_F(S_i, T_i) - e_E(S_i, T_i) \right| \leq 10 \sqrt{d \log d}\cdot \sqrt{|S_i| |T_i|}$ for $i = 1, \ldots, k$, and so
  \begin{align*}
    \left|2\cdot e_F(S,T) - e_E(S,T) \right|  & = \left|\sum_{i=1}^k \left( 2\cdot e_F(S_i, T_i) - e_E(S_i, T_i)\right)\right|
    \leq \sum_{i=1}^k \left| 2\cdot e_F(S_i, T_i) - e_E(S_i, T_i)\right|  \\
    &\leq 10 \sqrt{d \log d}\cdot \sum_{i=1}^k\sqrt{|S_i| |T_i|}
    \leq  10 \sqrt{d\log d}\cdot \sqrt{|S| |T|}\,,
  \end{align*}
  It is thus sufficient to prove the inequalities for those disjoint vertex sets $S, T$ that induce a connected subgraph $G[S\cup T]$.  
  
  Suppose we select $F$ by including each edge $e\in E$ with probability $1/2$ independently of other edges. That is, in the notation of Theorem~\ref{thm:LLL_constructive}, we have that $\mathcal{P}$ consists of $|E|$ mutually independent variables $\{P_e\}_{e\in E}$, where $P_e$ indicates whether $e\in F$ and $\Pr[P_e] = 1/2$. Now for each $S$ and $T$ such that $G[S\cup T]$ is connected, let $A_{S,T}$ be the ``bad'' event that $\left|2\cdot |\delta_F(S,T)| - |\delta_E(S,T)| \right| > 10 \sqrt{d \log d}\cdot \sqrt{|S| |T|}$. Note that $|\delta_F(S,T)|$ is the sum of at most $d\sqrt{|S||T|}$ independent variables, attaining values $0$ and $1$, and that the expected value of $|\delta_F(S,T)|$ equals $|\delta_E(S,T)|/2$. Thus by the Chernoff inequality we get
  \begin{align*}
    \Pr[A_{S,T}]   < d^{-6 |S\cup T|} \,.
  \end{align*}
  Similarly, if we let $D_v$ denote the bad event that the degree constraint of $v$ is violated, i.e.,  $\left|2\cdot e_F(v) - e_E(v) \right| > 10 \sqrt{d\log d}$. Then 
  \begin{align*}
    \Pr[D_{v}]   < d^{-6} \,.
  \end{align*}

  To apply Theorem~\ref{thm:LLL_constructive}, we analyze the dependency graph on the events: 
  \begin{itemize}
    \item There is an edge between $A_{S,T}$ and $A_{S',T'}$ if $\vbl(A_{S,T}) \cap \vbl(A_{S',T'}) \neq \emptyset  \Leftrightarrow \delta_E(S,T) \cap \delta_E(S', T') \neq \emptyset$. 
    \item There is an edge between $A_{S,T}$ and $D_v$ if $\vbl(A_{S,T}) \cap \vbl(D_{v}) \neq \emptyset  \Leftrightarrow \delta_E(S,T) \cap \delta_E(v) \neq \emptyset$.
    \item There is an edge between $D_u$ and $D_v$ if $\vbl(D_{u}) \cap \vbl(D_{v}) \neq \emptyset  \Leftrightarrow \delta_E(u) \cap \delta_E(v) \neq \emptyset$.
  \end{itemize}

      Consider now a fixed event $A_{S,T}$ and let $k = |S\cup T|$. We  bound the number of neighbors, $A_{S', T'}$, of $A_{S,T}$ with $|S'\cup T'| = \ell$.  Since we are interested in only subsets $S',T'$ such that $G[S' \cup T']$ is connected, this is bounded by the number of distinct subtrees on $\ell$ vertices in the associated graph $G$, with a root in one of the endpoints of an edge in $\delta(S,T)$. There are thus at most $2d\min(|S|, |T|)\leq dk$ many choices of the root and, as $G$ has degree at most $d$, the number of such trees is known to be at most  (see e.g.~\cite{Frieze1999})
\begin{align}
  \label{eq:nrcomponents2}
  dk \cdot {d (\ell-1) \choose \ell-1} \leq  dk \cdot (ed)^{\ell-1}\,,
\end{align}
where we used that ${d  (\ell-1) \choose \ell-1} \leq (ed)^{\ell-1}$.
Moreover, it is easy to see that $A_{S,T}$ has at most $2d \min (|S|, |T|)\leq dk$ neighbors $B_v$. 

Now to verify condition~\eqref{eq:cond} of Theorem~\ref{thm:LLL_constructive}, we  set $x(A_{S\cup T}) = d^{-3|S\cup T|}$ for every bad event $A_{S,T}$ and $x(B_v) = d^{-3}$ for every bad event $B_v$. So if we consider an event $A_{S,T}$ with  $k = |S\cup T|$, then 
\begin{align*}
  x(A_{S,T}) &\prod_{(S',T'): A_{S,T} \sim A_{S',T'}} \left( 1- x(A_{S',T'}) \right) \prod_{v: A_{S,T} \sim B_v} \left( 1- x(B_v)\right) \\ &\geq  d^{-3k} \prod_{\ell=1}^n \left(1- d^{-3\ell} \right)^{d k (ed)^{\ell-1}} \cdot \left( 1- d^{-3}\right)^{dk}  \\
  & \geq d^{-3k} \exp(- 2dk \sum_{\ell=1}^n d^{-3\ell} (ed)^{\ell-1} - 2dk d^{-3}) \\
  & \geq d^{-3k} e^{-3k}  > d^{-6k}/2 > \Pr[A_{S,T}]/2\,,
\end{align*}
where we used that $d$ is a sufficiently large constant, which is without loss of generality since if $d\leq 10 \sqrt{d\log (d)}$ then the lemma becomes trivial.
In other words,~\eqref{eq:cond} is satisfied for events $A_{S,T}$ with $\epsilon$ set to $1/2$. Let us now consider an event $B_v$. Clearly there is at most $d$  other events $B_u$ such that $B_u \sim B_v$. Moreover, by the same arguments as above there are at most $2d (ed)^{\ell-1}$ neighbors $A_{S, T}$ such that $|S\cup T| = \ell$. Hence
\begin{align*}
  x(B_v) & \prod \prod_{(S,T): B_v \sim A_{S,T}} \left( 1- x(A_{S,T}) \right)\prod_{u: B_v \sim B_u} \left( 1- x(B_u)\right)\\
  & \geq d^{-3} \prod_{\ell=1}^n \left(1- d^{-3\ell} \right)^{d  (ed)^{\ell-1}} \cdot \left( 1- d^{-3}\right)^d \\
  & \geq d^{-6}/2 \geq \Pr[B_v]/2\,,
\end{align*}
where the second to last inequality follows because of the same simplifications as done above with $k=1$. We have thus verified that~\eqref{eq:cond} is satisfied for all events  with $\epsilon$ set to $1/2$.

It remains to define an efficiently verifiable core subset $\mathcal{A}'\subseteq \mathcal{A}$ such that $1- \sum_{A \in \mathcal{A} \setminus \mathcal{A}'} x(A) \geq 1-n^{-3}$. We let
\begin{align*}
  \mathcal{A'} = \{B_v\}_{v\in V} \cup \{A_{S,T}\in \mathcal{A}: |S\cup T| \leq s\} \mbox{ where $s= \log_{d}(n)$}.
\end{align*}
By the same arguments as in~\eqref{eq:nrcomponents2},  there is at most $n \cdot {d  (\ell-1) \choose \ell-1} \leq n (ed)^{\ell-1}$ many vertex sets $U$ such that $|U| = \ell$ and $G[U]$ is connected. Moreover, for each $U$ there are $2^\ell$ possible ways of partitioning it into $S$ and $T$. Therefore, the following properties hold:
\begin{enumerate}
  \item $\mathcal{A}'$ is efficiently verifiable since it contains $n \cdot \sum_{\ell=1}^s (ed)^{\ell-1} 2^\ell = O(n \cdot (ed\cdot 2)^s)  = O(n^3)$ many events $A_{S,T}$ that can be efficiently enumerated by first selecting a vertex $r$ among $n$ choices, then considering all possible  trees rooted at $r$ with $\ell\leq s$ vertices, and all possible ways of partitioning such a component into $S$ and $T$. Moreover, the remaining $n$ events $B_v$ in $\mathcal{A}'$ contains $n$  are easy to verify in polynomial time.
  \item We have
    \begin{align*}
      \sum_{A_{S,T} \in \mathcal{A}\setminus \mathcal{A'}} x(A_{S,T}) & \leq \sum_{\ell = s+1}^n d^{-6\ell} \cdot(n \cdot (ed\cdot 2)^\ell ) \\
      & \leq n \cdot \sum_{\ell = s+1}^n d^{-4\ell}  \leq n d^{-4s}  = n^{-3}\,,
    \end{align*}
    where for the first inequality we again used that $d$ is a sufficiently large constant.
\end{enumerate}
We have verified Condition~\eqref{eq:cond} of Theorem~\ref{thm:LLL_constructive} and we have defined an efficiently verifiable core subset $\mathcal{A}'$ such that  $\sum_{A_S \in \mathcal{A}\setminus \mathcal{A}'} x(A_S) \leq n^{-3}$ and so the lemma follows.

\end{proof}

\cutsparsifieradv*
\begin{proof}
  Starting with $G$ we apply Lemma~\ref{lemma:cut_iteration_adv} $k$ times to obtain $\tilde{G}$. Let $F_i$ denote the edge set and let $d_i$ denote the maximum degree after round $i$. So $F_0 = E$ and $d_0=d$.
  By the guarantees of Lemma~\ref{lemma:cut_iteration_adv}, we have that with probability $1-n^{-3}$
\begin{align} 
  \label{eq:degree2}
  |2 d_{i+1} - d_i | &\leq 10 \sqrt{d_i \log (d_i)} 
\end{align}
and
\begin{align}
  \label{eq:set2}
  \left|2\cdot e_{F_{i+1}}(S,T) - e_{F_i}(S,T) \right| &\leq 10 \sqrt{d_i \log (d_i)}\cdot \sqrt{|S| |T|} \qquad \mbox{ for every disjoint $S,T\subseteq V$.}
\end{align}

As we apply Lemma~\ref{lemma:cut_iteration_adv} $k$ times with $k \leq \log(n)$, the union bound implies that the above inequalities are true for all invocations of that lemma with probability at least $1 - k\cdot n^{-3} \geq 1-n^{-2}$. From now on we assume that the above inequalities hold and show that the conclusion of the statement is always true in that case.
Specifically, we now prove by induction on $k$ that 
\begin{align*} 
  |2^k d_{k} - d_0 | &\leq  \eps d_0\mbox{, and} \\
  \left|2^k\cdot e_{F_{k}}(S,T) - e_{F_0}(S,T) \right| &\leq \eps d_0 \cdot \sqrt{|S||T|} \qquad \mbox{ for every disjoint $S,T\subseteq V$.}
\end{align*}
The claim holds trivially for $k=0$.
Assume it holds for all $i<k$, which in particular implies $2^id_i \leq 2 d_0$ for all $i <k$. By the triangle inequality and \eqref{eq:degree2}, 
\begin{eqnarray*}
|2^{k} d_k - d_0  |  & \leq &    \sum_{i=0}^{k-1}  | 2^{i} (2d_{i+1} - d_{i})| \leq 10  \sum_{i=0}^{k-1}  2^{i} \sqrt{d_i \log (d_i )}  \\
& \leq &  
10 \sum_{i=0}^{k-1}  2^{i}   \sqrt{2 (d_0/2^i) \log (2(d_0/2^i) )} \qquad \text{(induction hypothesis on $d_i$).}
\end{eqnarray*}
As the terms increase geometrically in $i$, this sum is $O( 2^k\sqrt{(d_0/2^k) \log ((d_0/2^{k}))}$ which is $\eps d_0$ by our assumption on $k$ and selection of $c$.

Finally, we note that $\left|2^k\cdot e_{F_{k}}(S,T) - e_{F_0}(S,T) \right| \leq \eps d_0 \cdot \sqrt{|S||T|}$  follows by the same calculations (using~\eqref{eq:set2} instead of~\eqref{eq:degree2}). 
\end{proof}

\bibliographystyle{plain}
\bibliography{refr}

\end{document}